\documentclass[journal,twoside,web,hidelinks,pdflatex]{ieeecolor}
\usepackage{generic}
\usepackage{cite}
\usepackage{amsmath,amssymb,amsfonts}

\usepackage{amsthm}
\usepackage{algorithmic}
\usepackage{graphicx}
\usepackage{adjustbox}
\usepackage{subcaption}
\usepackage{textcomp}
\usepackage{siunitx}
\usepackage{array} 
\usepackage{booktabs}
\usepackage{hyperref}
\usepackage{tikz}
\usetikzlibrary{arrows, arrows.meta, positioning, calc, math, patterns, patterns.meta,arrows.meta}
\usetikzlibrary{decorations.pathmorphing}
\usetikzlibrary{decorations.pathreplacing}

\newcommand{\sR}{\mathbb{R}}
\newcommand{\sN}{\mathbb{N}}
\newcommand{\ths}{\theta^{\star}}
\newcommand{\Ps}{P^{\star}}
\newcommand{\thhat}{\hat{\theta}}
\newcommand{\vd}{v_{\mathrm{d}}}
\newcommand{\vc}{v_{\mathrm{c}}}
\newcommand{\zd}{z_{\mathrm{d}}}
\newcommand{\zc}{z_{\mathrm{c}}}
\newcommand{\Tf}{T_{\mathrm{f}}}
\newcommand{\uf}{u_{\mathrm{f}}}
\newcommand{\yf}{y_{\mathrm{f}}}
\newcommand{\ef}{e_{\mathrm{f}}}

\newtheorem{remark}{Remark}
\newtheorem{fact}{Fact}

\DeclareMathOperator*{\argmin}{arg\,min}

\DeclareMathOperator{\Cov}{Cov}
\DeclareMathOperator{\Tr}{Tr}
\DeclareMathOperator{\cond}{cond}

\newtheorem{lemma}{Lemma}

\newtheorem{theorem}{Theorem}

\newcommand{\fs}{f_{\mathrm{S}}}
\begin{document}

\title{Fast Sampling for System Identification: Overcoming Noise, Offsets, and Closed-Loop Challenges with State Variable Filter}

\author{Ichiro Maruta, \IEEEmembership{Member, IEEE}, and Toshiharu Sugie, \IEEEmembership{Fellow, IEEE}
\thanks{This work has been submitted to the IEEE for possible publication. Copyright may be transferred without notice, after which this version may no longer be accessible.}
\thanks{This work was supported by JSPS KAKENHI Grant number 24K00908 and JST K Program Japan Grant Number JP-MJKP23G1.}
\thanks{Ichiro Maruta is with the Department of Aeronautics and Astronautics, Engineering,
        Kyoto University, Kyoto, 615-8540, Japan
        (e-mail: maruta@kuaero.kyoto-u.ac.jp).}
\thanks{Toshiharu Sugie is with the Department of Mechanical Engineering, Osaka University, 
Suita, Osaka 565-0871, Japan (e-mail: sugie@mech.eng.osaka-u.ac.jp).}}

\maketitle

\begin{abstract}
This paper investigates the effects of setting the sampling frequency significantly higher than conventional guidelines in system identification. 
Although continuous-time identification methods resolve the numerical difficulties encountered in discrete-time approaches when employing fast sampling (e.g., the problems caused by all poles approaching unity), the potential benefits of using sampling frequencies that far exceed traditional rules like the ``ten times the bandwidth'' guideline remained largely unexplored.
We show that using a state variable filter (SVF)-like least squares approach, the variance of the estimation error scales as $O(h)$ with the sampling interval $h$.
Importantly, this scaling holds even with colored noise or noise correlations between variables.
Thus, increasing the sampling frequency and applying the SVF method offers a novel solution for challenging problems such as closed-loop system identification and measurements with offsets.
Theoretical findings are supported by numerical examples, including the closed-loop identification of unstable multi-input multi-output (MIMO) systems. 
\end{abstract}


\section{Introduction}
\IEEEPARstart{S}{ystem} identification, the process of modeling dynamical systems from observed input-output data, is fundamental for advanced control, monitoring, and fault diagnosis. 
A key question in designing experiments for system identification is how to choose the data sampling interval.
The classical guideline of ``ten times the bandwidth of the system'' \cite[p.~452]{Ljung:1999} has long served as a useful rule of thumb for practitioners. 
This heuristic can be traced back to the early work by \cite{Astrom:1969}.
The primary reason for not setting the sampling frequency high, as mentioned in \cite{Ljung:1999}, is a numerical problem arising from all poles of the discrete-time model approaching $1$ when the sampling frequency is increased. 
This problem is indeed severe, as demonstrated by previous attempts to address the problem using the delta operator approach \cite{Middleton:1986}. 
In particular, when a target system exhibits multiple dynamics with significantly different time constants, no single appropriate sampling interval may exist. 

One natural way to circumvent this problem is to apply continuous-time system identification, namely, obtain a continuous-time plant model directly from the sampled I/O data \cite{Garnier:2008,GARNIER:2012}.
In fact, this strategy works well. However, the benefits of choosing a sampling frequency that greatly exceeds conventional guidelines have not been analyzed in detail, and the quantitative relation between the identification performance and the sampling frequency is not clear at all.

Therefore, first, we quantitatively evaluate the benefits of further reducing the sampling interval in this paper. 
In doing so, we demonstrate that the variance of the parameter estimation error scales as $O(h)$ with respect to the sampling interval $h$ and emphasize that this property is independent of factors such as noise color or correlation with the input.

This improvement in performance could provide a new avenue for addressing system identification problems that are important but have traditionally been challenging. 
One such problem is closed-loop system identification.  
In fact, closed-loop system identification is a critical yet challenging task in many practical applications, such as networked control and large-scale interconnected systems, where feedback induces correlations between inputs and measurement noise \cite{VS95, FORSSELL:1999, Ljung:1999}. 
Within the prediction error framework, approaches are typically classified as direct, indirect, or joint input-output methods \cite{FORSSELL:1999}. 
Direct methods use only I/O data and can accommodate high-order or nonlinear controllers, but they require precise noise modeling \cite{FL00}. 
Indirect approaches—exemplified by the Dual-Youla method—transform the closed-loop problem into an open-loop one, though they depend on accurate controller knowledge and may yield higher-order models \cite{HFK89, AGV11, MS21}.
Joint input-output techniques attempt to identify both the plant and controller simultaneously, generally under linearity assumptions.
In parallel developments, within the subspace identification framework, some methods have been extended to closed-loop data \cite{LM96, Qin06, VdV13, Tanaka:2021}.
Variants include joint input-output subspace approaches \cite{Ver93, KKP05}, innovation estimation methods \cite{QL03}, and two-stage procedures \cite{OOF06, KT07, Oku:2021}.
Direct methods such as SSARX and PBSID overcome some limitations by first estimating high-order ARX models and then applying model reduction \cite{Jan03, CP05}. 
More recent work has investigated automatic model order selection via nuclear norm minimization—a convex relaxation for rank minimization \cite{FHB01, FPST13, LHV13, Smi14, VH16} --- although its heavy computational load and sensitivity to low signal-to-noise ratios remain concerns \cite{NW16}.
Despite these extensive efforts, closed-loop system identification continues to face fundamental challenges, particularly in handling unstable plants or nonlinear controllers under heavy noise circumstances. 

Hence, we give a simple least squares method that overcomes all of the above problems.
More precisely, we show that a well-known State Variable Filter (SVF) approach combined with fast sampling exhibits quite powerful performance for tough closed-loop identification, where both input and output are contaminated by noises/disturbances and signal-to-noise ratio could be quite low (e.g., less than \SI{0}{dB}).

This paper is organized as follows. 
In Section~\ref{sec:fast_sampling_effect}, we analyze the benefits of fast sampling theoretically. 
Section~\ref{sec:svf} describes the identification problem with its solution based on the SVF approach, along with its extension to multi-input multi-output (MIMO) systems. 
In Section~\ref{sec:num_exp}, we demonstrate the effect of the fast-sampling via numerical experiments covering various scenarios such as closed-loop operation, unstable plants, and systems with measurement offsets. 
Finally, Section~\ref{sec:conclusions} concludes the paper.

\section{Fast-Sampling Effect}\label{sec:fast_sampling_effect}

\subsection{Fundamental Analysis via Power Spectral Density}
Let us first examine the quantitative relationship between sampling frequency and signal-to-noise ratio (SNR) by analyzing how the power spectral density (PSD) of noise changes with different sampling rates.
Consider a noise signal $v(t)$ with PSD $P_v(f)$. 
When sampled at frequency $\fs$, the PSD of the sampled signal $\vd(k)$ is given by:
\begin{align}
P_{\vd}(f) = \sum_{n = -\infty}^{\infty} P_v(f - n \fs).
\end{align}
This equation shows that frequency components above the Nyquist frequency $\fs/2$ are aliased into the frequency range $[-\fs/2, \fs/2]$, becoming indistinguishable from the original low-frequency content.
To prevent this aliasing effect from degrading SNR, system identification applications typically employ anti-aliasing filters that eliminate frequency components above $\fs/2$ before sampling.
This practical solution has become so standard that the presence of an anti-aliasing filter is often assumed when discussing sampling frequency selection in system identification theory, leading to guidelines such as the ``ten times the bandwidth'' rule.

However, implementing ideal anti-aliasing filters is often impractical or impossible in many real-world applications.
For instance, when measuring state variables in mechanical systems through discrete position encoders, or obtaining measurements through periodic physical sampling in chemical processes, implementing pre-sampling filters becomes impractical.

In these situations where anti-aliasing filters are absent, high-frequency components are aliased into the frequency range of interest as illustrated in the top part of Fig.~\ref{fig:aliasing_reduction}. But, increasing the sampling frequency reduces the amount of aliased noise as shown in the middle and bottom parts of 
Fig.~\ref{fig:aliasing_reduction}. 
This demonstrates that faster sampling could potentially improve SNR in the frequency range of interest.

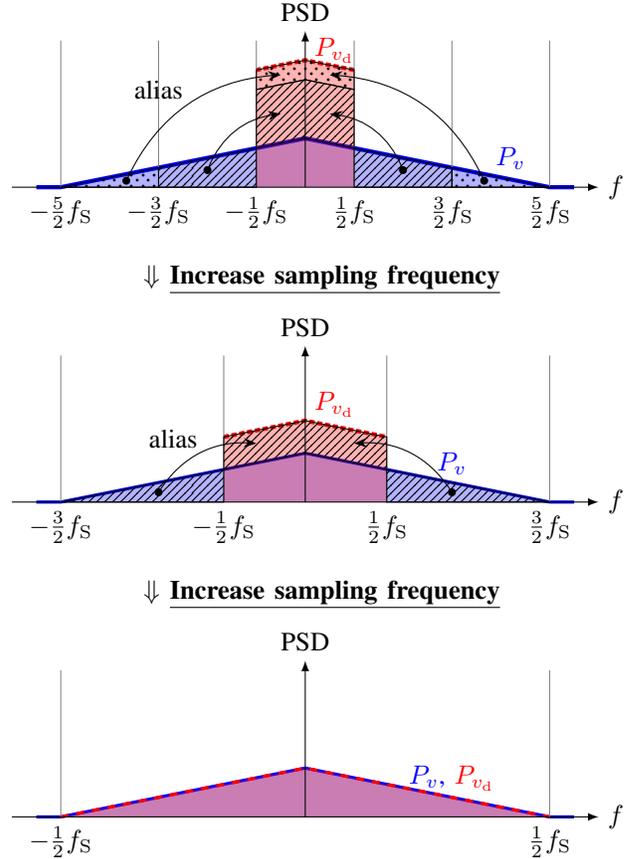
\begin{figure}[t]
    \centering
\begin{tikzpicture}[scale=0.65,domain=-5:5]%
\def\fsn{2}
\draw[ultra thick, blue, fill=blue, fill opacity = 0.3] (-5.5,0) -- (-5,0) -- (0, 1) -- (5,0) -- (5.5,0);
\node[blue] at (4.2,0.6) {$P_v$};
\draw[draw=none,fill=red, fill opacity = 0.3] (1,0) -- (1,2.4) -- (0, 2.6) -- (-1,2.4) -- (-1,0) ; 
\draw[ultra thick, red, dash pattern = on 2pt off 1pt] (1,2.4) -- (0, 2.6) -- (-1,2.4); 
\node[red] at (0.6,2.8) {$P_{\vd}$};
 
\draw[-latex] (-6,0) -- (6,0) node[right] {$f$};%
\draw[-latex] (0,0) -- (0,3.2) node[above] {PSD};%

\foreach \n in {-5,-3,-1,1,3,5} {
    \draw[very thin,color=gray] (\n*\fsn/2,0) -- (\n*\fsn/2,3); 
    \tikzmath{
      \nabs = int(abs(\n));
    }
    \ifthenelse{\n < 0}
        {\node[below] at (\n*\fsn/2,0) {$-\frac{\nabs}{2} \fs$};} 
        {\node[below] at (\n*\fsn/2,0) {$\frac{\n}{2} \fs$};}    
}

\draw[pattern=north east lines] (-3,0) -- (-1,0) -- (-1, 0.8) -- (-3,0.4) -- cycle; 
\draw[pattern=north east lines] (3,0) -- (1,0) -- (1, 0.8) -- (3,0.4) -- cycle; 
\draw[pattern=north east lines] (1,0.8) -- (1,2) -- (0, 2.2) -- (-1,2) -- (-1,0.8) -- (0,1) -- cycle; 

\draw[{Circle[length=3pt]}-{Stealth},shorten <=-1.5pt] (-2,0.7/2) to[bend left] (-0.5,1.5);
\draw[{Circle[length=3pt]}-{Stealth},shorten <=-1.5pt] (2,0.7/2) to[bend right] (0.5,1.5);
\node at (-3,2) {alias};

\draw[pattern={Dots[angle=45,distance={3pt}]}] (-5,0) -- (-3,0) -- (-3, 0.4) -- (-5,0) -- cycle; 
\draw[pattern={Dots[angle=45,distance={3pt}]}] (5,0) -- (3,0) -- (3, 0.4) -- (5,0) -- cycle; 
\draw[pattern={Dots[angle=45,distance={3pt}]}] (1,2) -- (1,2.4) -- (0, 2.6) -- (-1,2.4) -- (-1,2) -- (0,2.2) -- cycle; 

\draw[{Circle[length=3pt]}-{Stealth},shorten <=-1.5pt] (-3.66,0.4*1/3) to[bend left] (-0.5,2.3);
\draw[{Circle[length=3pt]}-{Stealth},shorten <=-1.5pt] (3.66,0.4*1/3) to[bend right] (0.5,2.3);
\end{tikzpicture}
\\[2mm]
\mbox{\textbf{$\Downarrow$ \underline{Increase sampling frequency}}}
\\[2mm]
\begin{tikzpicture}[scale=0.65,domain=-5:5]%
\def\fsn{5/1.5}
\draw[very thick, blue, fill=blue, fill opacity = 0.3] (-5.5,0) -- (-5,0) -- (0, 1) -- (5,0) -- (5.5,0);
\node[blue] at (3,0.8) {$P_v$};
\draw[draw=none,fill=red, fill opacity = 0.3] (-\fsn/2,0) -- (-\fsn/2,2/3+2/3) -- (0, 1+2/3) -- (\fsn/2,2/3+2/3) -- (\fsn/2,0); 
\draw[ultra thick, red, dash pattern = on 2pt off 1pt]  (-\fsn/2,2/3+2/3) -- (0, 1+2/3) -- (\fsn/2,2/3+2/3); 
\node[red] at (0.6,2) {$P_{\vd}$};

\draw[-latex] (-6,0) -- (6,0) node[right] {$f$};%
\draw[-latex] (0,0) -- (0,3.2) node[above] {PSD};%

\foreach \n in {-5,-5/3,5/3,5} {
    \draw[very thin,color=gray] (\n,0) -- (\n,3); %
    \tikzmath{
      \nabs = int(abs(\n*3.01/5));
      \nn = int(\n);
    }
    \ifthenelse{\nn < 0}
        {\node[below] at (\n,0) {$-\frac{\nabs}{2} \fs$};} %
        {\node[below] at (\n,0) {$\frac{\nabs}{2} \fs$};}    %
}

\draw[pattern=north east lines] (-\fsn*3/2,0) -- (-\fsn/2,0) -- (-\fsn/2, 2/3) -- cycle; 
\draw[pattern=north east lines] (\fsn*3/2,0) -- (\fsn/2,0) -- (\fsn/2, 2/3) -- cycle; 
\draw[pattern=north east lines] (-\fsn/2,2/3) -- (-\fsn/2,2/3+2/3) -- (0, 1+2/3) -- (\fsn/2,2/3+2/3) -- (\fsn/2,2/3) -- (0,1) -- cycle; 

\draw[{Circle[length=3pt]}-{Stealth},shorten <=-1.5pt] (-3,0.2) to[bend left] (-1,1.2);
\draw[{Circle[length=3pt]}-{Stealth},shorten <=-1.5pt] (3,0.2) to[bend right] (1,1.2);
\node at (-2.7,1.3) {alias};

\end{tikzpicture}   
\\[2mm]
\mbox{\textbf{$\Downarrow$ \underline{Increase sampling frequency}}}
\\[2mm]
\begin{tikzpicture}[scale=0.65,domain=-5:5]%
\def\fsn{5/1.5}
\draw[very thick, blue, fill=blue, fill opacity = 0.3] (-5.5,0) -- (-5,0) -- (0, 1) -- (5,0) -- (5.5,0);
\node at (3,0.8) {\color{blue}{$P_v$}, \color{red}{$P_{\vd}$}};
\draw[very thick, red, dashed, fill=red, fill opacity = 0.3] (-5,0) -- (0, 1) -- (5,0);
 
\draw[-latex] (-6,0) -- (6,0) node[right] {$f$};%
\draw[-latex] (0,0) -- (0,3.2) node[above] {PSD};%

\foreach \n in {-5,5} {
    \draw[very thin,color=gray] (\n,0) -- (\n,3); %
    \tikzmath{
      \nabs = int(abs(\n*1.01/5));
      \nn = int(\n);
    }
    \ifthenelse{\nn < 0}
        {\node[below] at (\n,0) {$-\frac{\nabs}{2} \fs$};} %
        {\node[below] at (\n,0) {$\frac{\nabs}{2} \fs$};}    %
}

\end{tikzpicture}

\caption{Reduction of aliased noise by increasing the sampling frequency.}
\label{fig:aliasing_reduction}
\end{figure}

\subsection{Quantitative Analysis of Noise Reduction through Low-pass Filtering with Fast Sampling}

To demonstrate how sampling frequency affects noise reduction quantitatively, let us first analyze the effect of low-pass filtering combined with fast sampling in case of white noise.

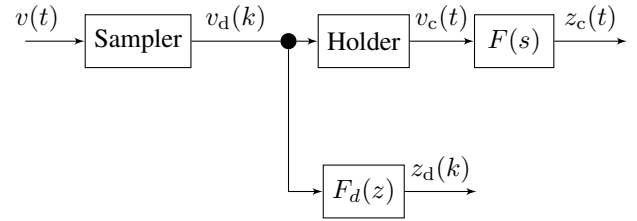
\begin{figure}[t]   %
 \centering
\tikzstyle{block} = [draw, rectangle, minimum height=2em, minimum width=3em]
\tikzstyle{sum} = [draw, circle, node distance=1cm]
\tikzstyle{branch} = [draw, circle, node distance=1.5cm, minimum size=0.2cm,inner sep=0pt,fill]
\tikzstyle{input} = [coordinate]
\tikzstyle{output} = [coordinate]
\tikzstyle{pinstyle} = [pin edge={to-,thin,black}]
\begin{tikzpicture}[auto, node distance=2cm,>=latex']
    \node [input, name=input] {};
    \node [block, right of=input, node distance=1.5cm] (sampler) {Sampler};
    \node [branch, right of=sampler, node distance=2cm] (branch) {};
    \node [block, right of=branch,  node distance=1cm] (holder) {Holder};
    \node [block, right of=holder, node distance=2cm] (Fs) {$F(s)$};
    \node [output, right of=Fs, node distance=1.5cm] (zc) {};

    \node [block, below of=holder, node distance=2cm] (Fdz) {$F_d(z)$};
    \node [output, right of=Fdz, node distance=1.5cm] (zd) {};

    \draw [->] (input) -- node[pos=0.2] {$v(t)$} (sampler);
    \draw [-] (sampler) -- node {$\vd(k)$} (branch);
    \node [input, above of=branch] {$\vd(k)$};
    \draw [->] (branch) --  (holder);
    \draw [->] (holder) -- node {$\vc(t)$} (Fs);
    \draw [->] (Fs) -- node {$\zc(t)$} (zc);

    \draw [->] (branch) |-  (Fdz);
    \draw [->] (Fdz) -- node {$\zd(k)$} (zd);
\end{tikzpicture}
 \caption{Discrete-time filter $F_h(z)$ with sampling interval $h$}
 \label{fig:Fh}
\end{figure}  %
Consider a continuous-time filter $F(s)$ that is stable and strictly proper, with input $\vc(t)$ and output $\zc(t)$.
Let $F_h(z)$ be its discretized version via zero-order hold (ZOH) with sampling interval $h$, where the input-output relationship is given by:
\begin{align}
\vc(t)  & = \vd(k) ~~(kh \le t < (k+1)h ) \\
\zd(k)  & = \zc(kh)  
\end{align}
as shown in Fig.~\ref{fig:Fh}.

Let $(A, B, C)$ be a minimal realization of $F(s)$.
Then, a state-space representation of $F_h(z)$ is  given by $(e^{Ah},~ \Phi (h) B,~ C)$ with $\Phi(h) =\int_{0}^{h} e^{A(h-\tau)} d \tau$.  
Using this state-space representation, the $H_2$ norms can be calculated as:
\begin{align}
\| F(s) \|_2 & = \sqrt{ \int_{0}^{\infty} (C e^{At} B)^2 dt } \\
\| F_h (z) \|_2 & = \sqrt{ \sum_{k=0}^{\infty} (C e^{Akh} \Phi(h) B)^2 }
\end{align}

Then, the following fact is well-known.

\begin{fact}[Filter Output Variance]\label{fact:filt_var}
Suppose the input $\vd(k)$ be zero-mean white noise satisfying
\begin{align}
{\bf E} [\vd^2(k)] = 1.
\label{eq:vd}
\end{align}
Then the corresponding output $\zd(k)$ is zero mean and its variance is given by
\begin{align}
 {\bf E} [\zd^2(k)]  =  \|F_h(z) \|_2^2
\label{eq:zd_variance}
\end{align}
\end{fact}
The right hand side of (\ref{eq:zd_variance}) can be evaluated by the following Lemma.

\begin{lemma}\label{lem:norm}
When $h \ll 1$ holds, we have
 \begin{align}
\| F_h(z) \|_2^2 \approx h \|F(s) \|_2^2.
\label{eq:H2normapproximation}
\end{align}

\begin{proof}
When $h\ll 1$, we have $\Phi(h) \approx h I$ and
\begin{align}
\| F_h(z)\|_2^2   \approx \sum_{k=0}^{\infty} (C e^{Akh} B)^2 h^2.
\label{eq:Fz2norm}
\end{align}

On the other hand, 
\begin{align}
h\| F(s)\|_2^2  = h\lim_{h \to 0}  \sum_{k=0}^{\infty}  (C e^{Akh} B)^2 h  \approx  \| F_h(z)\|_2^2,
\label{eq:Fs2norm}
\end{align}
and we obtain (\ref{eq:H2normapproximation}).
\end{proof}
\end{lemma}

The most important point here is that $\| F_h(z) \|_2$ approaches $0$ as $h \to 0$.
This, combined with Fact~\ref{fact:filt_var}, implies  that the variance of the filter output $\zd(k)$ decreases in proportion to the sampling interval $h$.
Meanwhile, when the filter input is a signal of interest and $h$ is sufficiently small for its dynamics, its magnitude remains invariant with respect to $h$. 
As a result, if a very small $h$ is chosen, the SNR improves drastically
by using any low-pass filter. 

To illustrate this effect, we provide a numerical example in Fig.~\ref{fig:filtering_effect} showing the filtering results for different sampling periods. 
In  Fig.~\ref{fig:filtering_effect}, the left hand side shows an example of filter input $\vd(k)$ (signal contaminated with noise) sampled with the interval $h$, and the right hand side corresponds the set of filter outputs $\zd(k)$ for different realizations of noise.  
The input $\vd(k)$ is injected to the filter $F(s)=\frac{1}{s+1}$ with ZOH.
The figure demonstrates how faster sampling combined with low-pass filtering effectively reduces noise while preserving the signal of interest.
\begin{figure*}[ht]   %
 \centering
 \begin{tikzpicture}
    \pgfmathsetmacro{\ws}{4.5}
    \tikzstyle{block} = [draw, rectangle, minimum height=2em, minimum width=3em]
    \tikzstyle{sum} = [draw, circle, node distance=1cm]
    \tikzstyle{input} = [coordinate]
    \tikzstyle{output} = [coordinate]
    
    \begin{scope}[yshift=1.5cm]
        \node[left] at (-7.5,0) {$h=\SI{20}{ms}$};
        \begin{scope}[xshift=-{\ws}cm]
            \node[inner sep=0pt] (in20ms) {\includegraphics[scale=1]{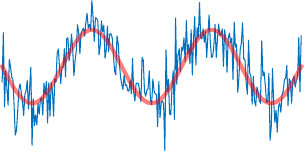}};
            \draw[-latex,red, thick] (-0.2,1) node[above] {$\text{signal}$} -- (-0.45,0.1) ;
            \draw[-latex,blue,thick] (-0.9,-1) node[below] {$\text{signal}+\text{noise}$} -- (-1.7,-0.4) ;        
        \end{scope}
        \node [block] (F) {$\displaystyle\frac{1}{s+1}$};            
        \node [input, name=input, left of=F] {};
        \node [output, name=output, right of=F] {};    
    
        \draw [->] (input) --  (F);
        \draw [->] (F) -- (output);
        \node at (\ws,0) {\includegraphics[scale=1]{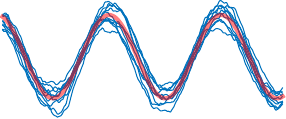}};
    \end{scope}
    \begin{scope}[yshift=-1.5cm]
        \node[left] at (-7.5,0) {$h=\SI{1}{ms}$};    
        \node at (-\ws,0) {\includegraphics[scale=1]{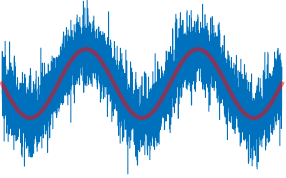}};
        \node [block] (F) {$\displaystyle\frac{1}{s+1}$};            
        \node [input, name=input, left of=F] {};
        \node [output, name=output, right of=F] {};    
    
        \draw [->] (input) --  (F);
        \draw [->] (F) -- (output);        
        \begin{scope}[xshift=-{-\ws}cm]
        \node (out1ms) {\includegraphics[scale=1]{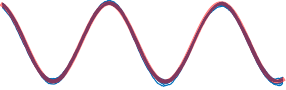}};
        \draw[decorate, decoration={brace, amplitude=7pt}] (2.4,-1.5) -- (-2.4,-1.5) node[midway, yshift=-0.6cm]{Filtered};
        \node at (0,-2.5) {(10 different realizations of noise)};
        \end{scope}
    \end{scope}
\end{tikzpicture}
 \caption{Filtering effect with fast sampling}
 \label{fig:filtering_effect}
\end{figure*}
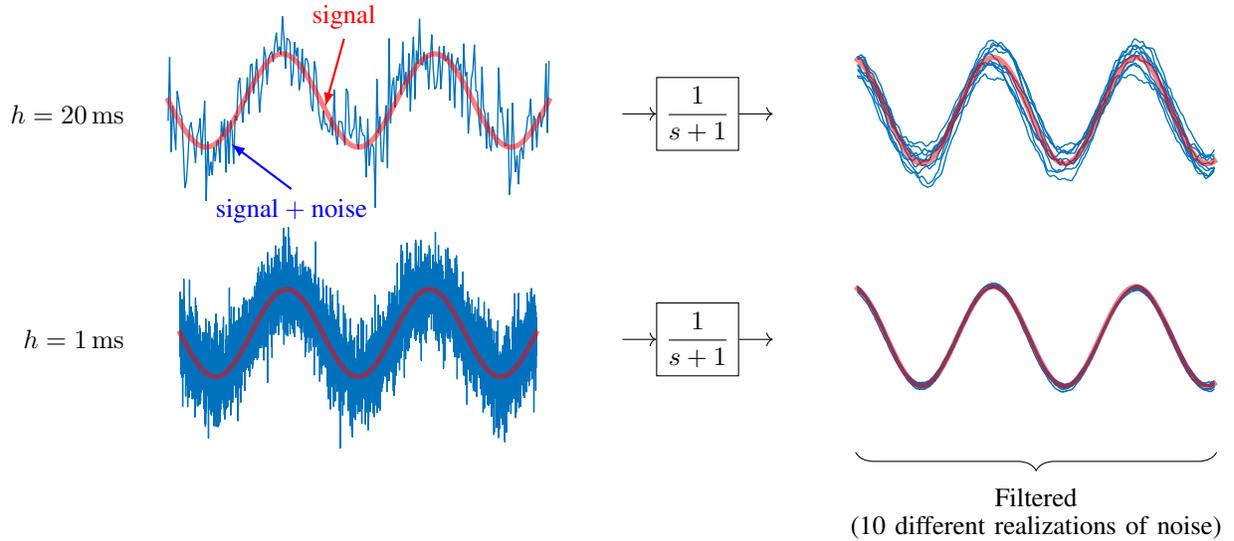   %

\subsection{Covariance of the Least Squares Estimators under Fast Sampling}
Building on the filtering analysis above, we now present a fundamental result regarding the asymptotic behavior of the covariance of the least squares estimator as $h \to 0$.

\begin{theorem}[Asymptotic Covariance Scaling with Fast Sampling]\label{thm:asymptotic_covariance_scaling}
Consider the linear regression model
\begin{align}
Y_h = \Phi_h\ths + V_h
\end{align}
where
\begin{align}
Y_h &=  [y(0), y(h), \ldots, y((N_h-1)h)]^\top\in \mathbb{R}^{N_h},\\
\Phi_h &= \left[\varphi(t)^\top, \varphi(h)^\top,\ldots, \varphi((N_h-1)h)^\top \right]^\top \in \mathbb{R}^{N_h \times n_\theta},\\
V_h &= [v(0), v(h), \ldots, v((N_h-1)h)]^\top\in \mathbb{R}^{N_h}
\end{align}
and, $y(t)\in\mathbb{R}$ is the observed output signal, $\varphi(t)\in\mathbb{R}^{n_\theta}$ is the regressor signal and assumed to be deterministic, $v(t) \in \mathbb{R}$ is the noise, and $\ths\in\mathbb{R}^{n_\theta}$ is the true parameter. 
Here, $N_h = \Tf/h+1$ is the number of samples for some fixed $\Tf>0$, and $n_\theta$ is the number of parameters.
Suppose the following conditions hold:

\begin{enumerate}
\item[(A1)] \textbf{Excitation Condition:}  
There exists a positive definite matrix $R\in \mathbb{R}^{n_\theta \times n_\theta}$ such that
\begin{align}
\frac{h}{\Tf}\Phi_h^\top \Phi_h \to R = \int_0^{\Tf}\varphi(t)^\top\varphi(t)\,\mathrm{d}t\quad \text{as } h \to 0.
\end{align}

\item[(A2)] \textbf{Total Energy of Noise:}  
Suppose $V_h$ is zero-mean with covariance $\Sigma_h \succeq 0$ (positive semidefinite) and 
\begin{align}
\Tr (\Sigma_h) = O(1) \quad\text{as} \quad h \to 0, 
\end{align}
where $\Tr (\cdot)$ denotes the trace of a matrix.
\end{enumerate}

Then the least squares estimator $\hat{\theta} = (\Phi_h^\top\Phi_h)^{-1}\Phi_h^\top Y_h$ satisfies
\begin{align}
\Tr(\Cov(\hat{\theta})) = O(h) \quad \text{as} \quad h \to 0.
\end{align}

\end{theorem}

\begin{proof}
From (A1),
\begin{align}
\left\|(\Phi_h^\top\Phi_h)^{-1}\right\| &= O(h),& \cond(\Phi_h^\top \Phi_h) &= O(1)
\end{align}
as $h\to 0$, where $\|\cdot\|$ denotes the spectral norm and $\cond(\Phi_h^\top \Phi_h) := \frac{\|\Phi_h^\top\Phi_h\|}{\|(\Phi_h^\top\Phi_h)^{-1}\|}$ denotes the condition number. 
Since $\|\Sigma_h\|\le\Tr(\Sigma_h)$, $\|\Sigma_h\|=O(1)$ as $h\to 0$ and
\begin{align}
\Cov(\thhat) &= (\Phi_h^\top\Phi_h)^{-1}\Phi_h^\top \Sigma_h \Phi_h (\Phi_h^\top \Phi_h)^{-1}\\
    \|\Cov(\thhat)\| &\le \|(\Phi_h^\top\Phi_h)^{-1}\|^2 \cdot \|\Phi_h^\top \Phi_h\| \cdot \|\Sigma_h\|\\
    &= \underbrace{\|(\Phi_h^\top\Phi_h)^{-1}\|}_{=O(h)} \cdot \underbrace{\cond(\Phi_h^\top \Phi_h)}_{=O(1)} \cdot \underbrace{\|\Sigma_h\|}_{=O(1)}\\
    &= O(h)
\end{align}
as $h\to 0$. 
And, the statement holds because $\Tr\left(\Cov(\thhat) \right) \le n_\theta \|\Cov(\thhat)\|$.

\end{proof}

\begin{remark}
At first glance, Theorem~\ref{thm:asymptotic_covariance_scaling} might applicable to both discrete-time and continuous-time models.
However, while Assumption (A1) is satisfied for continuous-time models, it does not hold in the discrete-time case.
In discrete-time models, the regressor vector $\varphi(t)$ is constructed from input-output signals that are delayed by one sampling interval.
As the sampling interval $h\to0$, these delayed signals converge to the identical signal, causing the regressor matrix $\Phi_h$ to lose full rank.
In contrast, for continuous-time models, particularly when employing methods like SVF, the regressor $\varphi(t)$ is typically constructed by passing the input-output signals through a continuous-time linear filter, which are independent of $h$.
Consequently, if these signals are linearly independent and continuous on almost every interval, the scaled matrix $\frac{h}{\Tf}\Phi_h^\top \Phi_h$ converges to a positive definite matrix $R$ as $h\to0$, satisfying Assumption (A1).
\end{remark}

\begin{remark}
    Assumption (A2) requires the total noise variance over the fixed interval $\Tf$, $\Tr\left(\Cov(\thhat) \right) = \sum_{k=0}^{N_h-1} E\left[v(kh)^2\right]$ to be bounded ($O(1)$) as $h\to0$.
    This is physically plausible without ideal anti-aliasing filters, where faster sampling captures roughly constant total noise energy, including aliased components (Fig.~\ref{fig:aliasing_reduction})
\end{remark}

Theorem~\ref{thm:asymptotic_covariance_scaling} demonstrates that if the regressor matrix $\Phi_h$ can be treated as deterministic, then the least squares parameter estimates converge asymptotically to the true parameter $\ths$ as $h\to0$. 
This convergence result holds independently of the noise properties (e.g., whether the noise is colored or exhibits correlation), provided that the total noise energy is bounded (as stated in Assumption (A2)).

\subsection{Impact of Sample-and-Hold Approximation Error}
A second, often overlooked benefit of fast sampling arises from the sample-and-hold mechanism itself.  
In many experimental setups, especially for the situations where the sampling frequency is lowered for system identification, both the control input to the plant and its measured output are subject to zero-order hold (ZOH) or similar sample-and-hold behavior.  
When the sampling interval $h$ is relatively large, the ZOH approximation induces a non-negligible error between the true continuous-time signal $x(t)$ and its held version $x_{\mathrm{ZOH}}(t)$, typically of order $O(h)$ in the supremum norm:
\begin{align}
    \sup_{t\in[ih,(i+1)h)} \bigl| x(t) - x_{\mathrm{ZOH}}(t)\bigr|
    \;\approx\;
    O(h).
\end{align}

\begin{remark}
Because the aliasing-based noise amplitude scales on the order of $O(\sqrt{h})$, the sample-and-hold error (which is $O(h)$) dominates when $h$ is relatively large, whereas aliasing-based noise eventually takes over as $h$ becomes small.    
\end{remark}

\section{System Identification using State Variable Filter}\label{sec:svf}
The State Variable Filter (SVF) is a classical approach for identifying continuous-time systems. 
Although it has been studied for many years, it has not been widely adopted for final parameter estimation due to its inherent bias, which stems from correlation between noise and regressors.
Nevertheless, SVF is often used to generate initial estimates for more sophisticated identification algorithms. 
In particular, it has been noted that SVF can handle challenging tasks such as closed-loop identification of unstable systems~\cite{CDC22}, where many other methods struggle or require additional information.

It is also important to note an often-overlooked source of error when implementing SVF in practice: the discrepancy introduced by computing filtered signals from sampled and held data. 
If the sampling interval is chosen to be relatively long (for instance, in accordance with the traditional ``ten times the bandwidth'' rule), the zero-order hold operation has a non-negligible effect on the filter's internal signals, resulting in degraded performance. 
Consequently, under these conventional guidelines for sampling frequency selection, SVF-based methods may appear less accurate than they actually are.

In contrast, \emph{fast sampling} provides an effective solution to these limitations. 
When the sampling interval is sufficiently small, the zero-order hold error becomes negligible, and the noise is substantially reduced in the filtered quantities. 
As a result, the bias in SVF estimates is mitigated, and the method's overall accuracy can improve dramatically. 

In the subsections that follow, we detail SVF algorithm. 
In particular, we show that the proposed procedure can handle offsets in measurement and manage unstable plants operating under closed-loop control without requiring knowledge of the controller. 
This flexibility makes the SVF approach a compelling choice for a wide range of identification problems.

\subsection{Problem Setting}

\begin{figure}[ht]
 \centering
\begin{tikzpicture}[auto, node distance=2cm, >=latex']
  \tikzstyle{block} = [draw, rectangle, minimum height=1.5em, minimum width=2em,node distance=1cm]
  \tikzstyle{sum} = [draw, circle, inner sep=0pt, minimum size=0.5em, node distance=1cm]
  \tikzstyle{split} = [draw, circle,fill, inner sep=0pt, minimum size=0.5em, node distance=1cm]
  \tikzstyle{input} = [node distance =1cm]
  \tikzstyle{output} = [node distance =1cm]

  \node [input, name=r_u] (r_u) {};
  \node [sum, right of=r_u] (sum_r_u) {};
  \node [sum, right of=sum_r_u] (sum_w) {};
  \node [block, right of=sum_w] (P) {$\Ps$};
  \draw [->] (r_u) -- node {$r_u$} (sum_r_u);
  \draw [->] (sum_r_u) -- node {$u$} (sum_w);
  \node[input, name=w, above of=sum_w,node distance=0.75cm](w){$w$};
  \draw [->] (w) -- (sum_w);  
  \draw [->] (sum_w) -- node {} (P);
  \node [sum, right of=P] (sum_eta) {};
  \draw [->] (P) -- node {} (sum_eta);
  \node[input, name=eta, above of=sum_eta, node distance=0.75cm](eta){$\eta$};
  \draw [->] (eta) -- (sum_eta);  
  \node [split, right of=sum_eta] (split_y) {};
  \draw [-] (sum_eta) -- node {} (split_y);
  \node [output, name=y, right of=split_y] {$y$};
  \draw [->] (split_y) -- (y);  
  \node [sum, above of=split_y, node distance=1.2cm](sum_r_y) {};
  \draw [->] (split_y) -- node[left, pos=0.8] {$-$} (sum_r_y);  
  \node [block, above of=P, node distance=1.2cm] (K) {$K$};
  \node[input, right of=sum_r_y](r_y){$r_y$};
  \draw [->] (r_y) -- (sum_r_y);  
  \draw [->] (sum_r_y) -- (K);  
  \draw [->] (K) -| (sum_r_u);  
  
\end{tikzpicture}
 \caption{The data generating system}
 \label{fig:data_generating_system}
\end{figure}
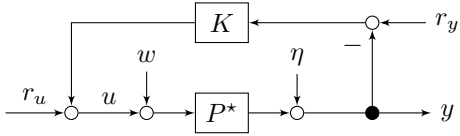

We consider a closed-loop system described by
the differential equation
\begin{align}
y(t)  &= \Ps(p) \left(u(t) +w(t)\right) + \eta(t), \\
u(t)  &= r_u(t) + K(p) \left(r_y(t)-y(t)\right),  
\end{align}
where $p$ denotes the time-domain differential operator and $\Ps(p)$ is the rational function of $p$ which represents the single-input single-output (SISO) plant to be identified, and  $K(p)$ represents an unknown controller which stabilizes $\Ps(p)$ (see Fig.~\ref{fig:data_generating_system}).
The scalar signals $u(t)$, $y(t)$, $r_u(t)$ and $r_y(t)$ are the plant input, the plant output measurement, the excitation signal and the reference signal, respectively.
Here, $r_u(t)$ and $r_y(t)$ may be unknown but are supposed to be rich enough for identification purpose.  
And $w(t)$ and $\eta(t)$ denote the input disturbance and the measurement noise, respectively, which may be colored and may have some offsets, namely 
\begin{align}
\boldmath{E}[w(t)] \ne 0, ~~
\boldmath{E}[\eta(t)] \ne 0.
\end{align}

The plant is assumed to belong to a known parametrized set
\[
{\cal P} := \{ P(p, \theta) ~|~ \theta \in \sR^{n_{\theta}} \},
\]
so it is described  by $\Ps (p)=P(p, \ths)$ for some parameter $\ths$. 
The I/O data are collected  for a given time period $t \in [0, \Tf]$, and sampling interval $h$ is chosen by the user.
Hence the available I/O data is given by
\begin{align}
{\cal  Z} =\{ u(kh), y(kh)  ~|~ k=0,1,\cdots, N_h-1 \}
\end{align}
where $N_h = (\Tf/h) + 1 \in \sN$ is the data length which depends on $h$.  
The problem here is to estimate $\ths$ from the I/O data alone.

\begin{remark}
When the plant $\Ps (p)$ itself is stable and $K(p) = 0$, the above problem becomes an ordinary open-loop identification problem. Namely, both open/closed-loop identification can be handled here.

Although the importance of identification in a closed-loop environment is widely recognized, there are not many that do not require information from controllers or external inputs.
In particular, there are few which discusses this type of closed-loop identification problems in continuous-time. 
In addition, the difficulty of identifying unstable systems is not widely recognized. 
This paper is tackling such a difficult problem. 
Hence, if we show that there is a way to handle such a problem without the need for specific technique, the contribution to users in industry would be enormous.
\end{remark}

\begin{remark}
In the control of mechanical systems, there is often a constant-valued (or deemed constant-valued) disturbance, such as friction, at the input channel. 
The need to cope with this type of offsets is very high in real-world control applications.  
In addition, there are many sensors whose offsets cannot be ignored when observing outputs. 
Once an offset is added, the effect on the identification accuracy is enormous. 
For this reason, we set up the problem in such a way that the offset exists in the disturbance or noise. 
It may be possible to estimate the magnitude of the offset in advance by assuming it.
However, in case of MIMO systems which will be discussed later, it does not seem realistic to estimate the amount of offset for each channel and correct it. 
In the following, we will discuss a method to remove the effect of the offset in identification without estimating the offset amount.
\end{remark}

\begin{remark}
This paper discusses a continuous-time model identification only.
However, once a continuous-time model is obtained, it is easy to obtain a discrete-time model for any sampling interval.    
\end{remark}

\subsection{Identification Procedure}

The SVF approach, which can be used to identify unstable systems and is simple to use, is employed here. 
In this section, for simplicity of notation, the model is expressed as follows:
\begin{align}
P(p,\theta) &=\frac{ \sum_{i=0}^{n-1} n_i p^i }{ p^n+ \sum_{i=0}^{n-1} d_i p^i } \\
 \theta   & = [d_{n-1}, \cdots, d_{0}, n_{n-1}, \cdots, n_0]^\top.
\end{align}

\subsection{SVF-like approach}

In the ideal case where noise-free continuous time I/O data $[u(t), y(t)], t \in [0, \Tf]$ is available,    
define
\begin{align}
e(t, \theta) & := \left(p^n+ \sum_{i=0}^{n-1} d_i p^i \right)y(t) - \left( \sum_{i=0}^{n-1} n_i p^i  \right)u(t)  \\
      & = y^{(n)}(t) +  \sum_{i=0}^{n-1}d_i y^{(i)}(t) - \sum_{i=0}^{n-1}n_i u^{(i)}(t),
\end{align}
then we will have 
\begin{align}
\hat{\theta} = \argmin_{\theta} \int_0^{\Tf}e^2(t,\theta) dt,
\label{eq:J}
\end{align}
which is equal to $\ths$.

Now, we have noise/disturbances and I/O data differentiation is not possible. Hence, we introduce a filter $F(p)$ whose relative-degree is larger or equal to the system order $n$, and calculate filtered I/O data $[\uf(t), \yf(t)]$ and their 
 higher order derivative can also be calculated by
\begin{align}
\yf(t) & =F(p)y(t), &\yf^{(k)}(t)&=p^kF(p)y(t),  \\ %
\uf(t) &=F(p)u(t), &\uf^{(k)}(t)&=p^{k} F(p)u(t) ~%
\end{align}
for ($1 \le k \le n$).
Define 
\begin{align}
\ef(t,\theta) & := \left( p^n+ \sum_{i=0}^{n-1} d_i p^i \right)\yf(t) - \left( \sum_{i=0}^{n-1} n_i p^i  \right)\uf(t)  \\
      & = \yf^{(n)}(t) +  \sum_{i=0}^{n-1}d_i \yf^{(i)}(t) - \sum_{i=0}^{n-1}n_i \uf^{(i)}(t), 
\label{eq:ef}
\end{align}
then $\hat{\theta}$ can be obtained from (\ref{eq:J}) by replacing $e(t, \theta)$ by $\ef(t,\theta)$.

Since $\ef(t, \theta)$ is affine in $\theta$, the optimization problem is a linear regression.

\begin{remark}[SVF-like approach]
In the conventional SVF framework, the order of the filter $F$ is set equal to the system order, and the main goal is to obtain the state variables from the filter. 
In reality, however, we can design the filter in various ways. 
As discussed later, for example, a filter with a band-pass characteristic can be used to cope with steady disturbances. 
In this paper, we refer to such cases as the ``SVF-like approach.'', where obtaining the state variables is not necessarily the primary objective. 
\end{remark}

\begin{remark}
The authors proposed a fixed pole observer method and showed its effectiveness in identifying stable/unstable plants in open/closed-loop environment\cite{CDC22}. 
And simple calculation shows that the fixed pole observer method is equivalent to SVF.
\end{remark}

We have (\ref{eq:ef}) at each sampled time $t=ih~(i=0,1, \cdots, N_h-1)$. 
This is equivalent to
\begin{align}
E &={Y}_A -D_F \theta, 
\label{eq:YA-DF} \\
E & :=  \begin{bmatrix}
e(t_0) \\
e(t_1) \\
\vdots \\
e(t_{N_h-1})
\end{bmatrix},~~
Y_A :=
\begin{bmatrix}
\yf^{(n)}(t_0) \\
\yf^{(n)}(t_1) \\
\vdots \\
\yf^{(n)}(t_{N_h-1}) \\
\end{bmatrix} \\
D_F &:=[-\yf, \uf] \\
\yf &:=
\begin{bmatrix}
\yf^{(n-1)}(t_0) & \cdots & \yf(t_0)  \\
\yf^{(n-1)}(t_1) & \cdots & \yf(t_1)  \\
\vdots               & \vdots & \vdots    \\
\yf^{(n-1)}(t_{N_h-1}) & \cdots & \yf(t_{N_h-1})  
\end{bmatrix} \\
\uf &:=
\begin{bmatrix}
 \uf^{(n-1)}(t_0) & \cdots & \uf(t_0) \\
 \uf^{(n-1)}(t_1) & \cdots & \uf(t_1) \\
 \vdots            & \vdots    & \vdots \\
 \uf^{(n-1)}(t_{N_h-1}) & \cdots & \uf(t_{N_h-1}) 
\end{bmatrix}
\end{align}
This yields the least squares solution as
\begin{align}
\hat{\theta} = (D_F^\top D_F)^{-1}D_F^\top Y_A.
\label{eq:LSsolution}
\end{align}

\subsection{On Accuracy of the Estimate}
Even when using the simple SVF-based method described above, one can achieve highly accurate parameter estimates by choosing a sufficiently small sampling interval $h$.

\subsubsection{Case of Zero-Mean Noise} \label{sec:accuracy_zero_mean} 
First, consider the situation where both the input disturbance $w(t)$ and the measurement noise $\eta(t)$ have zero means:
\begin{align}
{\bf E} [w] &=0,  & {\bf E} [\eta] &=0.
\end{align}
For simplicity, let $F(p)$ be a filter of relative degree at least $n+1$, where $n$ is the plant order.
Then each operator $p^{k}F(p)$ (for $k \le n$) is strictly proper.
Consequently, when we apply $p^{k}F(p)$ to the measured signal $y(t)$ and $u(t)$, the variance of noise components in the resulting filtered signals $\yf^{(k)}(t)$ and $\uf^{(k)}(t)$ become arbitrarily small as $O(h)$.
This is a result that can be derived from Lemma~\ref{lem:norm}, if we note that the components derived from $w(t)$ and $\eta(t)$ contained in $y(t)$ and $u(t)$ are the result of passing white noise through some linear filter.

The matrix $D_F$ is constructed from the filtered signals $\yf$ and $\uf$,
By Lemma~\ref{lem:norm}, the variance of the noise contained in these signals converges on the order of $O(h)$.
Meanwhile, if the component of $\frac{h}{\Tf} D_F ^\top D_F$ remains on the order of $O(1)$ as $h\to 0$ --- consistent with assumption (A1) in Theorem~\ref{thm:asymptotic_covariance_scaling} --- then, in the limit $h\to0$, we can treat $D_F$ as effectively deterministic.
Under these conditions, by letting $Y_h=Y_A$ and $\Phi_h=D_F$, Theorem~\ref{thm:asymptotic_covariance_scaling} implies that the estimation error variance in (\ref{eq:LSsolution}) is on the order of $O(h)$ as $h\to0$.

An important consequence of this fast-sampling effect is that one can identify the plant accurately whether it is stable or unstable, and whether or not it is operating in closed loop.
Furthermore, in the case of a closed-loop system,  information on the controller itself is not required at all, and it does not matter if the controller is nonlinear.
In other words, a very wide range of linear system identification problems can be handled in a unified manner 
(i.e., simple least-squares method alone) without the need to develop (or choose) any specific techniques.
 
\subsubsection{Case of Non-Zero-Mean Noise (Offsets)}\label{sec:offset}
Next, consider the practical scenario where the noise processes have constant offsets:
\begin{align}
{\bf E} [w] =:\bar{w} \neq  0,  \hspace
{3mm} {\bf E} [\eta] =: \bar{\eta} \neq 0.
\end{align}
A straightforward way to cancel these offsets is to include a differentiator $p$ in the filter as:
\begin{align}
F(p)=\frac{p n_F(p)}{d_F(p)}
\end{align}
where $d_F(p)$ is the denominator and $p  n_F(p)$ is the numerator.
Namely, the numerator includes the differential operator $p$. 
In fact, in this case, the filtered signal 
\begin{align}
\uf(t)=F(p)u(t)=F(p)\left\{ (u(t)-\bar{w}) +\bar{w}\right\}
\end{align}
consists of two parts, namely, $F(p)\left(u(t)- \bar{w}\right)$, and $F(p) \bar{w}$.
And, the second part vanishes as $t \to \infty$. 
So, the filtered signal converges to the response to the offset-removed signal $\left(u(t)- \bar{w}\right)$.
For example, when we choose  the filter as 
 \begin{align}
F(p)=\frac{p}{(p+1)^3},
\end{align}
$F(p) \bar{w}$ goes to $0$ as $t \to \infty$ and it is less then $\num{3.5e-5} \bar{w}$ for $t \ge \SI{15}{s}$. Therefore, it would be enough to construct 
$Y_{A}$ and $D_{F}$ based on the I/O data $[u(t), y(t)]$ after $t \ge 15$.
This is quite simple but effective as demonstrated later in numerical examples.  

\subsection{Extension to MIMO Systems}
While the above explanation focused on SISO plants for clarity, extending the SVF-like method to an MIMO plant is straightforward.
Suppose the plant has $m$ inputs and  $\ell$ outputs and let the model of the plant be given by
\begin{align}
P(p) &=\frac{N(p)}{d(p)} 
\end{align}
where $d(p)$ is a scalar polynomial of degree $n$, and $N(p)$ is an $\ell \times m$ matrix whose entries are polynomials in $p$ of degree at most $n-1$:
\begin{align}
d(p)  & = p^n+\sum_{k=0}^{n-1} d_{k} p^{k},&  d_k &\in \sR \\
\left( N(p)\right)_{ij} & = \sum_{k=0}^{n-1} N_{ij,k} p^{k}, & N_{ij,k} &\in \sR 
\end{align}
Define the parameter vector $\theta$ to include all coefficients in $d(p)$ and in every entry of $N(p)$:
\begin{multline}
\theta =[d_{n-1}, \cdots, d_{0}, N_{1 1, n-1}, \cdots, N_{11,0}, \\
               \cdots, N_{\ell m, n-1}, \cdots, N_{\ell m, 0}]^\top.
\end{multline}
The identification procedure is the same as before. 
Let
\begin{align}
u(t) &=[u_1(t), u_2(t), \cdots, u_m(t)]^\top, \\
y(t) &=[y_1(t), y_2(t), \cdots, y_{\ell}(t)]^\top
\end{align}
and calculate the filtered signal for each element as 
\begin{align}
u_{\mathrm{f},i} (t) &= F(p) u_i(t) ~~ (i=1, 2, \cdots, m) \\
y_{\mathrm{f},i} (t) & = F(p) y_i(t) ~~(i=1, 2, \cdots, \ell).
\end{align}
Defining 
\begin{align}
\uf(t) &= [u_{\mathrm{f},1}(t), u_{\mathrm{f},2}(t), \cdots, u_{\mathrm{f},m}(t)]^\top,  \\
\yf(t) &= [y_{\mathrm{f},1}(t), y_{\mathrm{f},2}(t), \cdots, y_{\mathrm{f},\ell}(t)]^\top, 
\end{align}
we have 
\begin{align}
\begin{split}
e(t) =\yf^{(n)}(t) +d_{n-1}\yf^{(n-1)}(t)+ \cdots + d_{0}\yf(t) \\
       - \{ N_{n-1} \uf^{(n-1)}(t) + \cdots + N_{0} \uf(t) \}
\end{split}
\end{align}
where
\begin{align}
\begin{split}
N_{k} := \begin{bmatrix}
N_{11,k} & N_{12,k} & \cdots & N_{1m,k} \\
N_{21,k}  & N_{22,k} & \cdots  &N_{2m,k}  \\
\vdots  &  \dots  & \ldots   & \vdots \\
N_{\ell 1, k} & N_{\ell 2, k} & \cdots & N_{\ell m, k}
\end{bmatrix}\\
(k=0,1, \cdots, n-1).
\end{split}
\end{align}
Let
\begin{align}
e(t)=[e_1(t), e_2(t), \cdots, e_{\ell}(t) ]^\top,
\end{align}
then each entry is given by
\begin{align}
\begin{split}
e_i(t) &= y_{\mathrm{f},i}^{(n)}(t)+[y_{\mathrm{f},i}^{(n-1)}(t), \cdots, y_{\mathrm{f},i}(t)][d_{(n-1)}, \cdots, d_0]^\top\\
        & - \sum_{j=1}^{m} [u_{\mathrm{f}, j}^{(n-1)}(t), \cdots, u_{\mathrm{f}, j}(t)][N_{ij,(n-1)}, \cdots, N_{ij,0}]^\top
\end{split}
\end{align}
Given the sampling time $t_k =kh$$(k=0,1, \cdots, N_h-1)$, we define   
\begin{align}
E_i  &:=
\begin{bmatrix}
e_i(t_0) \\
e_i(t_1) \\
\vdots  \\
e_i(t_{N_h-1})
\end{bmatrix},\quad
Y^{n}_{i}:=
\begin{bmatrix}
y_{\mathrm{f},i}^{(n)}(t_0) \\
y_{\mathrm{f},i}^{(n)}(t_1) \\
\vdots   \\
y_{\mathrm{f},i}^{(n)}(t_{N_h-1}) \\
\end{bmatrix},
\\
Y_{i} &:=
\begin{bmatrix}
y_{\mathrm{f},i}^{(n-1)}(t_0) & \cdots & y_{\mathrm{f},i}(t_0) \\
y_{\mathrm{f},i}^{(n-1)}(t_1) & \cdots & y_{\mathrm{f},i}(t_1) \\
\vdots  \\
y_{\mathrm{f},i}^{(n-1)}(t_{N_h-1}) & \cdots & y_{\mathrm{f},i}(t_{N_h-1}) \\
\end{bmatrix},
\\
U_{j}  &:=
\begin{bmatrix}
u_{\mathrm{f},j}^{(n-1)}(t_0) & \cdots & u_{\mathrm{f},j}(t_0) \\
u_{\mathrm{f},j}^{(n-1)}(t_1) & \cdots & u_{\mathrm{f},j}(t_1) \\
\vdots  \\
u_{\mathrm{f},j}^{(n-1)}(t_{N_h-1}) & \cdots & u_{\mathrm{f},j}(t_{N_h-1}) \\
\end{bmatrix},
\\
U &:=[U_1, U_2, \cdots, U_m]
\end{align}
then we have
\begin{align}
E_i  &= Y^n_i + Y_i [d_{n-1}, \cdots, d_0]^\top  \nonumber \\
      & ~~- \sum_{j=1}^{m} U_{j} [N_{ij,n-1}, \cdots, N_{ij,0}]^\top,
\end{align}
which can be represented by a matrix equation form by 
\begin{align}
\begin{bmatrix}
E_1 \\
E_2 \\
\vdots \\
E_{\ell}
\end{bmatrix}
=
\begin{bmatrix}
Y^n_1 \\
Y^n_2 \\
\vdots \\
Y^n_{\ell}
\end{bmatrix}
-
\begin{bmatrix}
-Y_1      & U      & 0        & \cdots  & 0 \\
-Y_2     &  0       & U       & 0        & 0   \\
\vdots & \vdots & 0       & \ddots & 0 \\
-Y_{\ell} & 0        & \cdots &  0       &U
\end{bmatrix}
\begin{bmatrix}
\theta_d \\
\theta_{N1} \\
\vdots \\
\theta_{N\ell}
\end{bmatrix}
\label{eq:Y-DFmimo}
\end{align}
with
\begin{align}
\theta_d &:=
\begin{bmatrix}
d_{n-1} & d_{n-2} & \cdots & d_0
\end{bmatrix}^\top
\in {\bf R}^{n}
\\
\theta_{Ni}
&:=
\begin{bmatrix}
[N_{i1, n-1},  N_{i1,n-2}, \cdots, N_{i1,0}]^\top \\
[N_{i2, n-1},  N_{i2,n-2}, \cdots, N_{i2,0}]^\top \\
\dots \\
[N_{im,n-1},  N_{im,n-2}, \cdots, N_{im,0}]^\top 
\end{bmatrix}
\in {\bf R}^{nm}
\end{align}
We can obtain the estimate $\hat{\theta}$ from \eqref{eq:LSsolution} by regarding \eqref{eq:Y-DFmimo} as
\eqref{eq:YA-DF}.

\subsection{Summary}

The key point is that the SVF method itself is simple and flexible, but it suffers from bias caused by both noise correlation and sample-and-hold errors at low sampling rates. 
By adopting a fast sampling strategy, we can eliminate these large errors, making the SVF method a practical tool for system identification. 
This allows us to identify plant models accurately without needing special techniques or detailed knowledge about the controller, even in challenging conditions.

In the next section, we present simulation results to demonstrate the effectiveness of using fast sampling with the SVF-like method and how it compares with conventional approaches.

\section{Numerical Experiments}\label{sec:num_exp}

In this section, we present a series of numerical experiments to demonstrate the effectiveness of:
\begin{enumerate}
    \item choosing a fast sampling interval beyond the conventional ``ten times the bandwidth'' rule, and 
    \item using the SVF-like identification method in that fast-sampling regime.
\end{enumerate}
We focus on closed-loop identification scenarios where the plant may be stable or unstable, and in some cases subject to constant offsets. 
All experiments illustrate how fast sampling reduces the noise effect and improves identification accuracy.

\subsection{Overview of the Experimental Setup}
We tested four plants: three SISO plants (\texttt{P1}, \texttt{P2}, \texttt{P3}) and one MIMO plant (\texttt{P4}). 
The details of these plants and their controllers are summarized in Table~\ref{tab:plants_controllers}. 

\begin{table*}[t]
\centering
\caption{Summary of Plants $P(s)$ and Controllers $K(s)$ Used in the Numerical Examples}
\label{tab:plants_controllers}
\renewcommand{\arraystretch}{1.1}
\begin{tabular}{@{}clll@{}}
\toprule
\textbf{System} & $P(s)$ & $K(s)$ &\textbf{Notable Features} \\
\midrule
\texttt{P1} 
& $\displaystyle \frac{1}{\,s - 1\,}$ 
& $\displaystyle \frac{3\,s + 7}{\,0.2\,s - 2\,}$ 
& 1st-order, unstable pole\\
\texttt{P1o} 
& Same as \texttt{P1}, but with offsets on noise
& Same as \texttt{P1}
& Tests offset handling \\
\texttt{P1f} 
& Same as \texttt{P1}, but noise-free
& Same as \texttt{P1}
& Tests the effects other than noise \\
\texttt{P2} 
& $\displaystyle \frac{s + 1}{\,s^2 \,+\, 0.5\,s \,+\, 1\,}$
& $\displaystyle \frac{0.5s-0.75}{s+1}$
& 2nd-order stable system \\
\texttt{P3} 
& $\displaystyle \frac{s - 1}{\,s^2 - 4\,}$
& $\displaystyle \frac{5.5s+11}{s-2.8}$
& 2nd-order with unstable pole and zero\\
\texttt{P4} 
& $\frac{\begin{bmatrix}s^2 + 1.4\,s + 4.4&1\\
3\,s^2 -3.8\,s +7.6 & s^2 + s -1
\end{bmatrix}}{s^3 - 0.6\,s^2 + 3.6\,s -4}$ 
& Observer-based regulator
& 2-input, 2-output, 3rd-order unstable system\\
\bottomrule
\end{tabular}
\end{table*}

In each experiment, the plant operates under feedback control with an unknown controller \(K(p)\) as in Fig.~\ref{fig:data_generating_system}, and we collect input-output data for a fixed duration \(\Tf=\SI{20}{s}\). 
Specifically:
\begin{itemize}
    \item \textbf{Sampling and Decimation:}
    We first simulate each system at a very fine step size (=\(\SI{5e-6}{s}\)) over the interval \([-10,20]\), and decimate the high-resolution data at various sampling intervals \(h\) (e.g., \(\SI{10}{ms}, \SI{1}{ms}, \SI{0.1}{ms}, \dots\)), effectively emulating different sampling frequencies while using the same dataset.
    
    For the SVF-like method, the decimated signal is held using a zero-order hold and then fed into a bank of filters; from the sampled outputs of these filters, $Y_A$ and $D_F$ are constructed.
    However, to eliminate the influence of initial conditions, the data corresponding to the time interval $[-10, 0)$ seconds are discarded.
    
    In contrast, the conventional method used for comparison is provided with the entire dataset covering the time interval $[-10, 20]$ seconds.
    This setup offers a somewhat favorable condition for the conventional approach.

    \item \textbf{Excitation Signals:} 
    We inject square-wave signals with unit amplitude and random phase into the reference inputs \(r_u(t)\) and \(r_y(t)\), which are not known for identification procedure. 
    This ensures sufficient excitation for identification.
    \item \textbf{Noise and Offsets:}
    In our simulations, both the input disturbance $w(t)$ and the measurement noise $\eta(t)$ are generated by sampling and holding the white Gaussian noise with standard deviation \(0.1\) in every \(\SI{5e-6}{s}\).
    For the offset case (\texttt{P1o}), an additional constant offset of \(\bar{w}=10\) is added to the input and \(\bar{\eta}=1\) to the output.
    For MIMO setting (\texttt{P4}), noise on each variable is independent from the others.

    \item \textbf{Metric for Comparison:} 
    Commonly used norms (e.g., $H_2$ or $H_\infty$ norms) are generally not finite for unstable systems, making it difficult to reliably quantify the discrepancy between the unstable true system and its model.
    To overcome this limitation, we use the Vinnicombe ($\nu$-gap) metric\cite{Vinnicombe:1993}.
    Briefly, the $\nu$-gap metric quantifies the distance between two systems via their normalized co-prime factorizations and is bounded between $0$ and $1$. 

    The $\nu$-gap metric possesses excellent properties for evaluating models obtained via closed-loop identification.
    In particular, $\delta_\nu(\hat{P},\Ps)$ corresponds to the degree of robustness required by a controller --- designed based on $\hat{P}$ --- to stabilize the actual system $\Ps$.
    Hence, it serves as a good indicator of the model's usefulness in control system design.
\end{itemize}

\subsection{SVF-like Continuous-Time Identification}
As discussed in Section~\ref{sec:svf}, we apply an SVF-like procedure. 
Specifically, for each sampled dataset \(\{u(kh),\,y(kh)\}\), we compute
\begin{align}
  \uf^{(\ell)}(t) &= p^\ell F(p)\,\bar{u}(t),
  &
  \yf^{(\ell)}(t) &= p^\ell F(p)\,\bar{y}(t),\\
  \bar{u}(t) &:= u(\lceil t/h \rceil\cdot h), & \bar{y}(t) &:= y(\lceil t/h \rceil\cdot h).
\end{align}
We then form the linear regression matrix and perform least squares to solve for the continuous-time parameters \(\hat{\theta}\). 

\subsubsection{Choice of the Filter \texorpdfstring{$F(s)$}{F(s)}}
For \texttt{P1-3}, the filter is
\begin{align}
  F(s) \;=\; \frac{s}{(s+1)\,\bigl(s^2 + 1.8\,s + 1\bigr)},
  \label{eq:F}
\end{align}
This yields proper operators \(p^\ell F(p)\) for $\ell\le2$.

The factor \(s\) in the numerator of \(F(s)\) is to nullify the offset in steady-state for the case offsets or constant biases are present.

\subsubsection{Comparison with Discrete-Time Methods}
For comparison, we also identify each system using a conventional discrete-time ARX (or SSARX) method on the same decimated datasets. 
We adopt the commands \texttt{arx} or \texttt{n4sid} with SSARX option in MATLAB System Identification Toolbox.
These methods generally require a model with higher order than that of the target system. 
Here, we set the order to ten, which produced relatively good results.

This comparison illustrates how traditional sampling frequency rule works for discrete-time methods and is not applicable for the continuous-time SVF-like method.

\subsection{Representative Setting (P1)}
In the representative problem setting (\texttt{P1}), the plant is given by 
\[
   P^\star(s) \;=\; \frac{1}{s - 1},
\]
operating in the closed loop with
\[
   K(s) = \frac{3\,s + 7}{0.2\,s - 2}.
\]
White noise is added to both input and output, and we also consider a biased version (\texttt{P1o}) with nonzero offsets on the noise and a noise-free version (\texttt{P1f}).
We collect \(\Tf=\SI{20}{s}\) of data and test sampling intervals \(h\) from \(\SI{100}{ms}\) (i.e.\ \SI{10}{Hz}) down to \(\SI{0.01}{ms}\) (i.e.\ \SI{100}{kHz}).

\paragraph{I/O Data Example}
The blue lines in Fig.~\ref{fig:io-example-p1} show an example of the measured input and output for \texttt{P1} with \(h = \SI{10}{ms}\). 
In the figure, red lines show the results without noise. 

\begin{figure}[!t]
\centering
\includegraphics[scale=1]{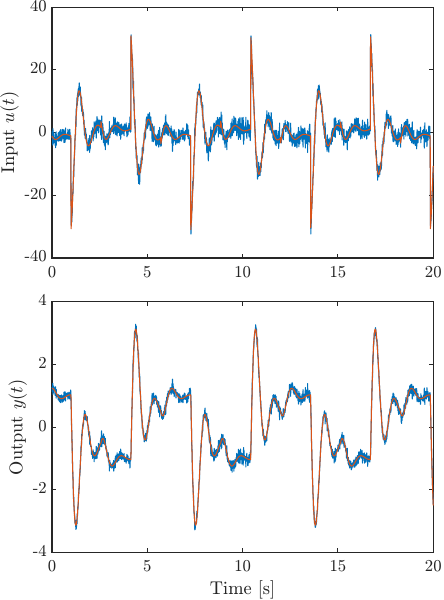}
\caption{Example I/O data for \texttt{P1} at sampling interval \(h = \SI{10}{ms}\). 
Offsets/biases are zero in this example.}
\label{fig:io-example-p1}
\end{figure}

\paragraph{Bode Diagram Comparisons}
Figure~\ref{fig:bode-p1} compares 50 identified models (from 50 different noise and random excitation phase realizations) for each sampling rate:   
we show the magnitude and phase of identified models obtained by 
using the SVF-like method and discrete-time ARX model identification. 

Increasing the sampling frequency improves the accuracy of both methods up to a certain degree (about \SI{10}{ms} in this example), but excessive increases cause numerical problems for discrete-time model-based methods. 
On the other hand, the SVF-like method can benefit from an increase in the number of samples to any extent.

This trend is common to all examples, and while there is a sweet spot in the choice of sampling frequency for discrete-time model-based methods, as has been conventionally known, and an unlimited increase in sampling frequency does yield poor results. This is not the case for continuous-time model-based methods.
\begin{figure*}[tbhp]
\centering
 \begin{tabular}{|>{\centering\arraybackslash}m{3ex}|>{\centering\arraybackslash}m{7.5cm}|>{\centering\arraybackslash}m{7.5cm}|}
 \hline
\rule{0pt}{3ex}& \textbf{SVF-like} & \textbf{Discrete-Time ARX} \\ \hline
\rule[-2.8cm]{0pt}{5.6cm}\raisebox{-.5\height}{\rotatebox{90}{$h=\SI{100}{ms}$}}& 
\raisebox{-2.8cm}{\includegraphics[scale=1]{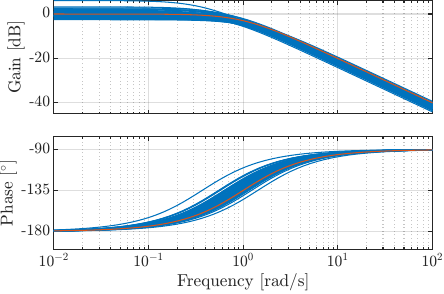}} &
\raisebox{-2.8cm}{\includegraphics[scale=1]{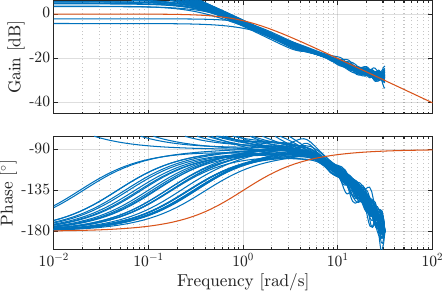}}\\ \hline
\rule[-2.8cm]{0pt}{5.6cm}\raisebox{-.5\height}{\rotatebox{90}{$h=\SI{10}{ms}$}}&
\raisebox{-2.8cm}{\includegraphics[scale=1]{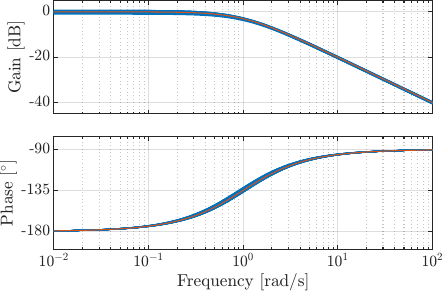}} &
\raisebox{-2.8cm}{\includegraphics[scale=1]{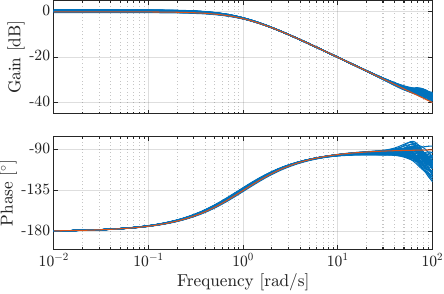}}\\ \hline
\rule[-2.8cm]{0pt}{5.6cm}\raisebox{-.5\height}{\rotatebox{90}{$h=\SI{1}{ms}$}}&
\raisebox{-2.8cm}{\includegraphics[scale=1]{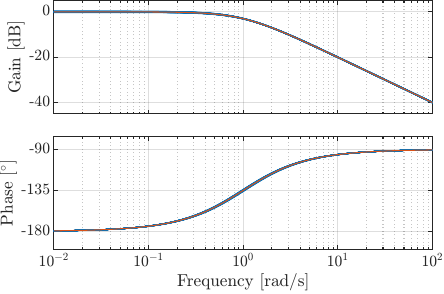}} &
\raisebox{-2.8cm}{\includegraphics[scale=1]{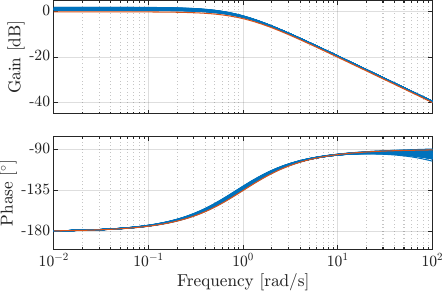}}\\ \hline
\rule[-2.8cm]{0pt}{5.6cm}\raisebox{-.5\height}{\rotatebox{90}{$h=\SI{0.1}{ms}$}}&
\raisebox{-2.8cm}{\includegraphics[scale=1]{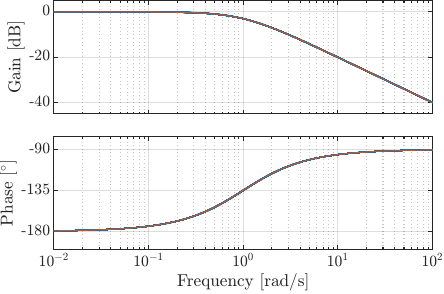}} &
\raisebox{-2.8cm}{\includegraphics[scale=1]{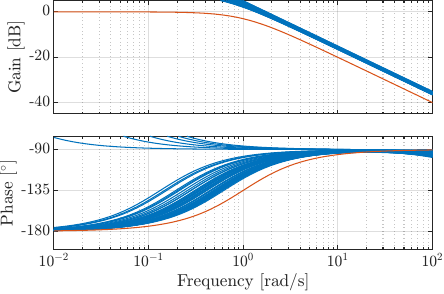}}\\\hline
\end{tabular}
\caption{Bode magnitude/phase of 50 identified models for \texttt{P1}. 
Blue lines: identified models; red line: true plant.}
\label{fig:bode-p1}
\end{figure*}

\paragraph{Convergence}

Also, statistics of parameter estimates in \texttt{P1} for various sampling intervals are summarized as box plots in 
Fig.~\ref{fig:boxplot-p1}, where $n_0$ and $d_0$ are coefficients of the numerator and the denominator of the plant model, respectively.
\begin{figure}
    \centering
    \includegraphics[scale=1]{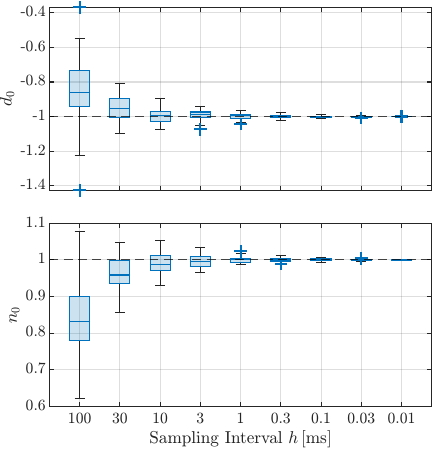}
    \caption{Statistics of parameter estimates in \texttt{P1} vs. sampling interval}
    \label{fig:boxplot-p1}
\end{figure}
In the figure, true values for each parameter is shown as dashed line, and the statistics of the estimates are shown as box plots.
As seen in the figure, the bias and variance of the parameter estimates vanish as $h \to 0$.

\paragraph{Effect of Offsets (\texttt{P1o})}
We also tested a variant where a constant bias is added to the noise signals. 
In this case, a constant disturbance is added and the input/output data are as shown in Fig.~\ref{fig:io-example-p1o}, but the other conditions are identical to \texttt{P1}. 
Fig.~\ref{fig:bode-p1o} shows the Bode diagram of the model obtained with a sampling frequency of $h=\SI{10}{ms}$, which is in the sweet spot. As can be seen from the figure, the usual identification method cannot handle unknown disturbances, while the SVF-like method can easily handle this type of problem by utilizing degrees of freedom in the filter.
\begin{figure}[!t]
\centering
\includegraphics[scale=1]{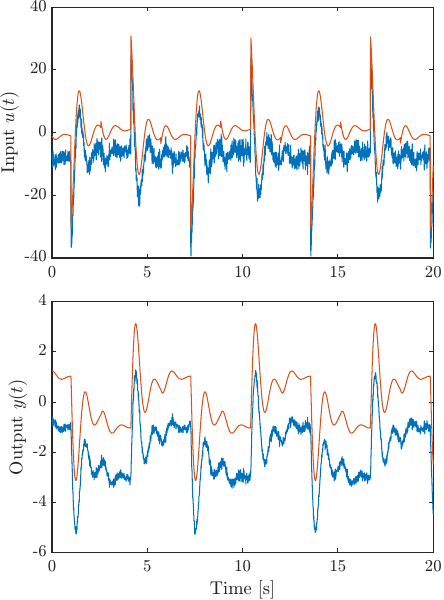}
\caption{Example I/O data for \texttt{P1o} at sampling interval \(h = \SI{10}{ms}\). 
Offsets/biases are $\bar{w}=10$ and $\bar{\eta}=1$ in this example.}
\label{fig:io-example-p1o}
\end{figure}

\begin{figure*}[t]
\centering
 \begin{tabular}{|>{\centering\arraybackslash}m{3ex}|>{\centering\arraybackslash}m{7.5cm}|>{\centering\arraybackslash}m{7.5cm}|}
 \hline
\rule{0pt}{3ex}& \textbf{SVF-like} & \textbf{Discrete-Time ARX} \\ \hline
\rule[-2.8cm]{0pt}{5.6cm}\raisebox{-.5\height}{\rotatebox{90}{$h=\SI{10}{ms}$}}&
\raisebox{-2.8cm}{\includegraphics[scale=1]{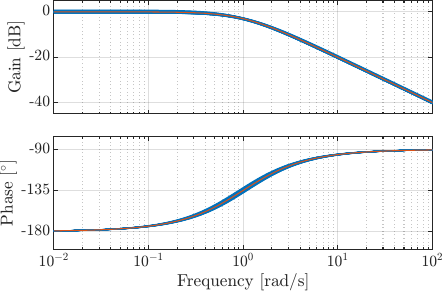}} &
\raisebox{-2.8cm}{\includegraphics[scale=1]{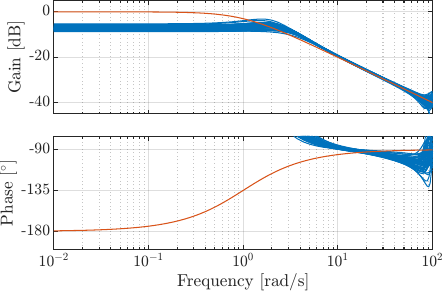}}\\ \hline
\end{tabular}
\caption{Bode plot of 50 identified models for the problem with offset (\texttt{P1o}).}
\label{fig:bode-p1o}
\end{figure*}

\subsection{MIMO Example (\texttt{P4})}
The plant in \texttt{P4} is a two-input/two-output unstable system stabilized by an observer-based controller, which places observer poles at $1.1\times\left(-1, -0.9 \pm \sqrt{0.9^2 - 1}\right)$ and state feedback poles at  $0.8\times\left(-1, -0.9 \pm \sqrt{0.9^2 - 1}\right)$.
Since \texttt{P4} requires the proper operator for $\ell=3$, the filter is changed to
\begin{align}
  F(s) \;=\; \frac{1}{(s+1)^2\left\{\left(s+0.2\right)^2 + 1.99^2\right\}}.
  \label{eq:F3}
\end{align}

Fig.~\ref{fig:mimo-bode} shows the Bode plot of the models obtained by the SVF-like method with fast-sampling $h=\SI{0.1}{ms}$. 
As seen in the figure, fast-sampling and simple SVF-like approach produce appropriate models even for a tough problem with the unstable MIMO system.

\begin{figure*}[t]
\centering
\includegraphics[scale=1]{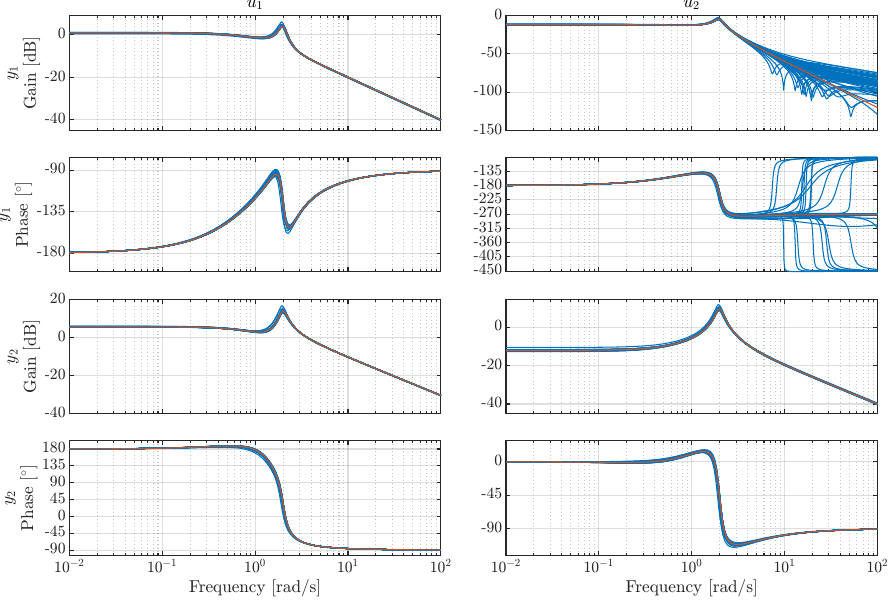}\\
\caption{Bode plot of 50 identified models \texttt{P4} at \(h=\SI{0.1}{ms}\) with SVF-like method. }
\label{fig:mimo-bode}
\end{figure*}

\subsection{Convergence and Variance Trends}
Fig.~\ref{fig:id-converge} summarizes the overall estimation error as a function of the sampling interval \(h\). Specifically, we plot
\[
  \frac{\|\hat{\theta}-\ths\|^2}{\|\ths\|^2}
\]
averaged over 50 noise and excitation signal realizations, with different settings. 
\texttt{P2} and \texttt{P3} are additional examples for investigating the generality of the result, and apart from the differences in the plants and controllers shown in Table~\ref{tab:plants_controllers}, there are no differences from \texttt{P1}.

In all cases, the convergence of the mean-square error of the parameters is generally in the range $O(h)$ to $O(h^2)$, and in the region where $h$ is relatively large, the $O(h^2)$ error based on sample holding becomes more dominant, while in the region where $h$ is sufficiently small, the $O(h)$ error due to folded noise becomes dominant.
The results \texttt{P1f}, which are based on noise-free measurements, show convergence of approximately $O(h^2)$ over the entire range, suggesting that the $O(h)$ factor is indeed due to aliasing noise.

\begin{figure}[t]
\centering
\includegraphics[scale=1]{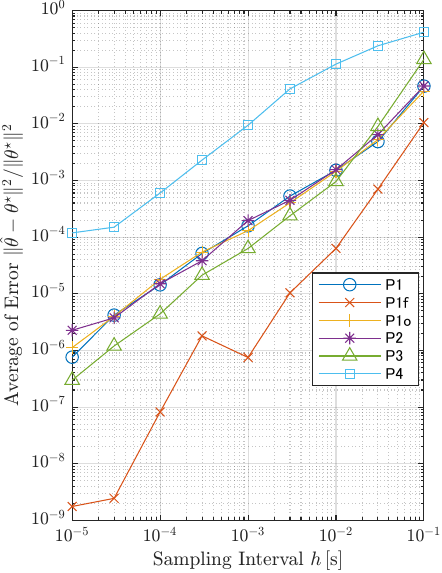}
\caption{Average normalized parameter error 
\(\|\hat{\theta}-\ths\|^2 / \|\ths\|^2\)
vs.\ sampling interval \(h\). }
\label{fig:id-converge}
\end{figure}

The results show that it is worth considering speeding up the sampling in system identification based on continuous-time models.

\begin{remark}
One might argue that shorter $h$ increases the number of I/O data and 
consequently improves the accuracy of the identified models.
However, this reasoning is incorrect. 
Longer data may improve the variance of identified models but does not improve the bias error. 
Shorter sampling interval is the essential factor contributing to improved accuracy 
(i.e., both smaller bias and smaller variance).
\end{remark}
\begin{remark}
It should be stressed that it is enough for us to solve just one least squares problem in the SVF-like method, and any special techniques or iterative procedures are not necessary at all. Hence, this identification method is quite easy to use for non-expert users. Nevertheless, it provides us a very accurate plant model no matter how heavy the measurement noise is, as long as faster sampling is available.
\end{remark}

\begin{figure*}[t]
\centering
\begin{tabular}{ccc}
\includegraphics[scale=1]{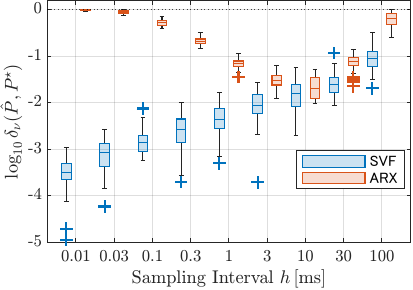}&
\includegraphics[scale=1]{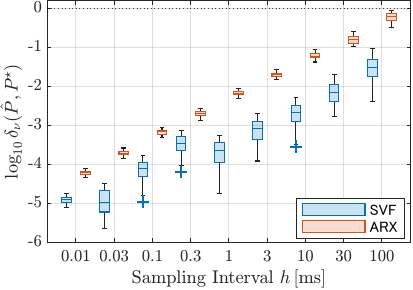}\\
(a) \texttt{P1}&
(b) \texttt{P1f}\\[3mm]
\includegraphics[scale=1]{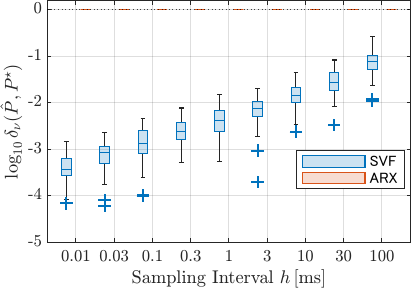}&
\includegraphics[scale=1]{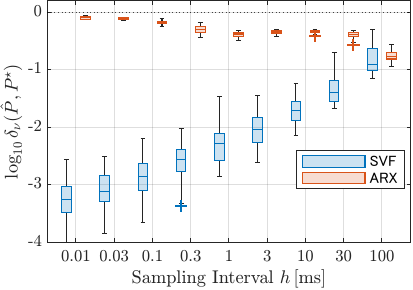}\\
(c) \texttt{P1o}&
(d) \texttt{P2}\\[3mm]
\includegraphics[scale=1]{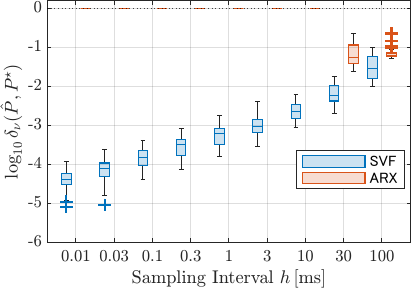}&
\includegraphics[scale=1]{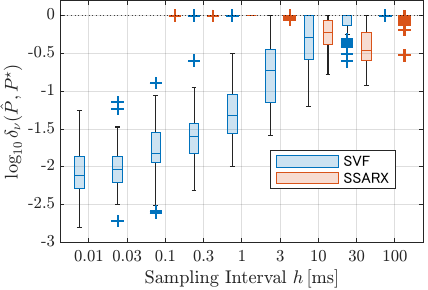}\\
(e) \texttt{P3}&
(f) \texttt{P4}
\end{tabular}
\caption{Distribution of the $\nu$-gap metric between the true plants and identified models across various sampling intervals ($h$).
Note that for the MIMO system \texttt{P4} (f), SSARX results are omitted for $h < \SI{0.1}{ms}$ due to the prohibitive computational cost at these high sampling frequencies.}
\label{fig:gap_metric}
\end{figure*}

\subsection{Comparison with methods based on discrete-time models}
Finally, Fig.~\ref{fig:gap_metric} shows the $\nu$-gap metric between the true plants and identified models for various sampling intervals. 
In the figure, distributions of the common logarithm of the $\nu$-gap are shown in box plots. 
The dotted line indicates $\delta_\nu(\hat{P},\Ps)=1$, which means that the model $\hat{P}$ is significantly different from the true plant $\Ps$ and that the stabilizing compensator for the model does not stabilize the true plant. 
Conversely, $\delta_\nu(\hat{P},\Ps)=0$ means that the model and the plant are perfectly matched.
The results confirm two key observations. 
First, SVF-like methods consistently improve as sampling becomes faster, with the $\nu$-gap metric approaching zero as $h$ decreases. Second, discrete-time models (ARX/SSARX) show the well-known performance degradation when sampling becomes too fast. 
These methods initially improve with faster sampling but then deteriorate as $h$ continues to decrease. 
This confirms the conventional sampling frequency selection rule for discrete-time methods and highlights the advantage of SVF-like approaches with fast sampling.

\section{Conclusions}\label{sec:conclusions}
In this paper, we revisited the conventional rule of choosing the sampling frequency to be ``ten times the bandwidth'' and showed that, under modern hardware capabilities, much faster sampling can substantially improve system identification accuracy.

We provided a theoretical basis for this claim by analyzing how the variance of the noise can be reduced via appropriate low-pass filters.
Based on this fact, it is shown that when parameters are estimated via the SVF-like approach using least squares, the variance of the estimation error scales as $O(h)$ with respect to the sampling interval $h$.
In addition, this scaling property is shown to be valid even with the presence of colored noise or the noise correlations between variables.

Thus, increasing the sampling frequency and applying the SVF-like method offers a novel solution for difficult problems such as closed-loop system identification.
To validate the universal effectiveness of the fast sampling approach, the numerical experiments including the tough situations, that is, closed-loop setting, unstable plant, offset in measurement, multi-input multi-output (MIMO) systems, are given. 

Although the scaling law presented here, where the variance of parameter estimates decreases proportionally to the sampling interval ($O(h)$), is expected to be applicable to a broader range of identification schemes and models, detailed analysis is left for future research.

\section*{References}
\bibliographystyle{IEEEtran}
\bibliography{refs}

\begin{thebibliography}{10}
\providecommand{\url}[1]{#1}
\csname url@samestyle\endcsname
\providecommand{\newblock}{\relax}
\providecommand{\bibinfo}[2]{#2}
\providecommand{\BIBentrySTDinterwordspacing}{\spaceskip=0pt\relax}
\providecommand{\BIBentryALTinterwordstretchfactor}{4}
\providecommand{\BIBentryALTinterwordspacing}{\spaceskip=\fontdimen2\font plus
\BIBentryALTinterwordstretchfactor\fontdimen3\font minus \fontdimen4\font\relax}
\providecommand{\BIBforeignlanguage}[2]{{%
\expandafter\ifx\csname l@#1\endcsname\relax
\typeout{** WARNING: IEEEtran.bst: No hyphenation pattern has been}%
\typeout{** loaded for the language `#1'. Using the pattern for}%
\typeout{** the default language instead.}%
\else
\language=\csname l@#1\endcsname
\fi
#2}}
\providecommand{\BIBdecl}{\relax}
\BIBdecl

\bibitem{Ljung:1999}
L.~Ljung, \emph{System Identification: Theory for the User}, ser. Prentice Hall information and system sciences series.\hskip 1em plus 0.5em minus 0.4em\relax Prentice Hall PTR, 1999.

\bibitem{Astrom:1969}
\BIBentryALTinterwordspacing
K.~Åström, ``On the choice of sampling rates in parametric identification of time series,'' \emph{Information Sciences}, vol.~1, no.~3, pp. 273--278, 1969. [Online]. Available: \url{https://www.sciencedirect.com/science/article/pii/S0020025569800137}
\BIBentrySTDinterwordspacing

\bibitem{Middleton:1986}
R.~Middleton and G.~Goodwin, ``Improved finite word length characteristics in digital control using delta operators,'' \emph{IEEE Transactions on Automatic Control}, vol.~31, no.~11, pp. 1015--1021, 1986.

\bibitem{Garnier:2008}
H.~Garnier and L.~Wang, \emph{Identification of Continuous-time Models from Sampled Data}, ser. Advances in Industrial Control.\hskip 1em plus 0.5em minus 0.4em\relax Springer London, 2008.

\bibitem{GARNIER:2012}
\BIBentryALTinterwordspacing
H.~Garnier and P.~Young, ``What does continuous-time model identification have to offer?'' \emph{IFAC Proceedings Volumes}, vol.~45, no.~16, pp. 810--815, 2012, 16th IFAC Symposium on System Identification. [Online]. Available: \url{https://www.sciencedirect.com/science/article/pii/S1474667015380538}
\BIBentrySTDinterwordspacing

\bibitem{VS95}
\BIBentryALTinterwordspacing
P.~M. {Van Den Hof} and R.~J. Schrama, ``Identification and control — closed-loop issues,'' \emph{Automatica}, vol.~31, no.~12, pp. 1751--1770, 1995, trends in System Identification. [Online]. Available: \url{https://www.sciencedirect.com/science/article/pii/000510989500094X}
\BIBentrySTDinterwordspacing

\bibitem{FORSSELL:1999}
\BIBentryALTinterwordspacing
U.~Forssell and L.~Ljung, ``Closed-loop identification revisited,'' \emph{Automatica}, vol.~35, no.~7, pp. 1215--1241, 1999. [Online]. Available: \url{https://www.sciencedirect.com/science/article/pii/S0005109899000229}
\BIBentrySTDinterwordspacing

\bibitem{FL00}
------, ``Identification of unstable systems using output error and {B}ox-{J}enkins model structures,'' \emph{IEEE Transactions on Automatic Control}, vol.~45, no.~1, pp. 137--141, 2000.

\bibitem{HFK89}
F.~Hansen, G.~Franklin, and R.~Kosut, ``Closed-loop identification via the fractional representation: Experiment design,'' in \emph{1989 American Control Conference}, 1989, pp. 1422--1427.

\bibitem{AGV11}
\BIBentryALTinterwordspacing
J.~C. Agüero, G.~C. Goodwin, and P.~M. {Van den Hof}, ``A virtual closed loop method for closed loop identification,'' \emph{Automatica}, vol.~47, no.~8, pp. 1626--1637, 2011. [Online]. Available: \url{https://www.sciencedirect.com/science/article/pii/S0005109811002500}
\BIBentrySTDinterwordspacing

\bibitem{MS21}
I.~Maruta and T.~Sugie, ``A simple framework for identifying dynamical systems in closed-loop,'' \emph{IEEE Access}, vol.~9, pp. 31\,441--31\,453, 2021.

\bibitem{LM96}
\BIBentryALTinterwordspacing
L.~Ljung and T.~McKelvey, ``Subspace identification from closed loop data,'' \emph{Signal Processing}, vol.~52, no.~2, pp. 209--215, 1996, subspace Methods, Part II: System Identification. [Online]. Available: \url{https://www.sciencedirect.com/science/article/pii/0165168496000540}
\BIBentrySTDinterwordspacing

\bibitem{Qin06}
\BIBentryALTinterwordspacing
S.~J. Qin, ``An overview of subspace identification,'' \emph{Computers \& Chemical Engineering}, vol.~30, no.~10, pp. 1502--1513, 2006, papers form Chemical Process Control VII. [Online]. Available: \url{https://www.sciencedirect.com/science/article/pii/S009813540600158X}
\BIBentrySTDinterwordspacing

\bibitem{VdV13}
\BIBentryALTinterwordspacing
G.~van~der Veen, J.-W. van Wingerden, M.~Bergamasco, M.~Lovera, and M.~Verhaegen, ``Closed-loop subspace identification methods: an overview,'' \emph{IET Control Theory \& Applications}, vol.~7, no.~10, pp. 1339--1358, 2013. [Online]. Available: \url{https://ietresearch.onlinelibrary.wiley.com/doi/abs/10.1049/iet-cta.2012.0653}
\BIBentrySTDinterwordspacing

\bibitem{Tanaka:2021}
\BIBentryALTinterwordspacing
H.~Tanaka and K.~Ikeda, ``Minimal realization of an unstable plant under feedback,'' \emph{IFAC-PapersOnLine}, vol.~54, no.~7, pp. 773--778, 2021, 19th IFAC Symposium on System Identification SYSID 2021. [Online]. Available: \url{https://www.sciencedirect.com/science/article/pii/S2405896321012295}
\BIBentrySTDinterwordspacing

\bibitem{Ver93}
\BIBentryALTinterwordspacing
M.~Verhaegen, ``Application of a subspace model identification technique to identify {LTI} systems operating in closed-loop,'' \emph{Automatica}, vol.~29, no.~4, pp. 1027--1040, 1993. [Online]. Available: \url{https://www.sciencedirect.com/science/article/pii/0005109893901042}
\BIBentrySTDinterwordspacing

\bibitem{KKP05}
\BIBentryALTinterwordspacing
T.~Katayama, H.~Kawauchi, and G.~Picci, ``Subspace identification of closed loop systems by the orthogonal decomposition method,'' \emph{Automatica}, vol.~41, no.~5, pp. 863--872, 2005. [Online]. Available: \url{https://www.sciencedirect.com/science/article/pii/S0005109805000075}
\BIBentrySTDinterwordspacing

\bibitem{QL03}
\BIBentryALTinterwordspacing
S.~J. Qin and L.~Ljung, ``Closed-loop subspace identification with innovation estimation,'' \emph{IFAC Proceedings Volumes}, vol.~36, no.~16, pp. 861--866, 2003, 13th IFAC Symposium on System Identification (SYSID 2003), Rotterdam, The Netherlands, 27-29 August, 2003. [Online]. Available: \url{https://www.sciencedirect.com/science/article/pii/S1474667017348681}
\BIBentrySTDinterwordspacing

\bibitem{OOF06}
H.~Oku, Y.~Ogura, and T.~Fujii, ``{MOESP}-type closed-loop subspace model identification method,'' \emph{Transactions of the Society of Instrument and Control Engineers}, vol.~42, no.~6, pp. 636--642, 2006.

\bibitem{KT07}
\BIBentryALTinterwordspacing
T.~Katayama and H.~Tanaka, ``An approach to closed-loop subspace identification by orthogonal decomposition,'' \emph{Automatica}, vol.~43, no.~9, pp. 1623--1630, 2007. [Online]. Available: \url{https://www.sciencedirect.com/science/article/pii/S0005109807001707}
\BIBentrySTDinterwordspacing

\bibitem{Oku:2021}
\BIBentryALTinterwordspacing
H.~Oku and K.~Ikeda, ``On consistency of output-error closed-loop subspace model identification for systems compensated by general {LTI} controllers,'' \emph{IFAC-PapersOnLine}, vol.~54, no.~7, pp. 767--772, 2021, 19th IFAC Symposium on System Identification SYSID 2021. [Online]. Available: \url{https://www.sciencedirect.com/science/article/pii/S2405896321012283}
\BIBentrySTDinterwordspacing

\bibitem{Jan03}
\BIBentryALTinterwordspacing
M.~Jansson, ``Subspace identification and {ARX} modeling,'' \emph{IFAC Proceedings Volumes}, vol.~36, no.~16, pp. 1585--1590, 2003, 13th IFAC Symposium on System Identification (SYSID 2003), Rotterdam, The Netherlands, 27-29 August, 2003. [Online]. Available: \url{https://www.sciencedirect.com/science/article/pii/S1474667017349868}
\BIBentrySTDinterwordspacing

\bibitem{CP05}
\BIBentryALTinterwordspacing
A.~Chiuso and G.~Picci, ``Consistency analysis of some closed-loop subspace identification methods,'' \emph{Automatica}, vol.~41, no.~3, pp. 377--391, 2005, data-Based Modelling and System Identification. [Online]. Available: \url{https://www.sciencedirect.com/science/article/pii/S000510980400319X}
\BIBentrySTDinterwordspacing

\bibitem{FHB01}
M.~Fazel, H.~Hindi, and S.~Boyd, ``A rank minimization heuristic with application to minimum order system approximation,'' in \emph{Proceedings of the 2001 American Control Conference. (Cat. No.01CH37148)}, vol.~6, 2001, pp. 4734--4739.

\bibitem{FPST13}
\BIBentryALTinterwordspacing
M.~Fazel, T.~K. Pong, D.~Sun, and P.~Tseng, ``Hankel matrix rank minimization with applications to system identification and realization,'' \emph{SIAM Journal on Matrix Analysis and Applications}, vol.~34, no.~3, pp. 946--977, 2013. [Online]. Available: \url{https://doi.org/10.1137/110853996}
\BIBentrySTDinterwordspacing

\bibitem{LHV13}
\BIBentryALTinterwordspacing
Z.~Liu, A.~Hansson, and L.~Vandenberghe, ``Nuclear norm system identification with missing inputs and outputs,'' \emph{Systems \& Control Letters}, vol.~62, no.~8, pp. 605--612, 2013. [Online]. Available: \url{https://www.sciencedirect.com/science/article/pii/S0167691113000819}
\BIBentrySTDinterwordspacing

\bibitem{Smi14}
R.~S. Smith, ``Frequency domain subspace identification using nuclear norm minimization and hankel matrix realizations,'' \emph{IEEE Transactions on Automatic Control}, vol.~59, no.~11, pp. 2886--2896, 2014.

\bibitem{VH16}
\BIBentryALTinterwordspacing
M.~Verhaegen and A.~Hansson, ``{N2SID}: Nuclear norm subspace identification of innovation models,'' \emph{Automatica}, vol.~72, pp. 57--63, 2016. [Online]. Available: \url{https://www.sciencedirect.com/science/article/pii/S0005109816302163}
\BIBentrySTDinterwordspacing

\bibitem{NW16}
S.~Navalkar and J.~van Wingerden, ``Nuclear norm-enhanced recursive subspace identification: Closed-loop estimation of rapid variations in system dynamics,'' in \emph{2016 American Control Conference ({ACC})}, 2016, pp. 936--941.

\bibitem{CDC22}
I.~Maruta and T.~Sugie, ``Direct closed-loop identification of continuous-time systems using fixed-pole observer model,'' in \emph{2022 IEEE 61st Conference on Decision and Control (CDC)}, 2022, pp. 2172--2177.

\bibitem{Vinnicombe:1993}
G.~Vinnicombe, ``Frequency domain uncertainty and the graph topology,'' \emph{IEEE Transactions on Automatic Control}, vol.~38, no.~9, pp. 1371--1383, 1993.

\end{thebibliography}

\begin{IEEEbiography}[{\includegraphics[width=1in,height=1.25in,clip,keepaspectratio]{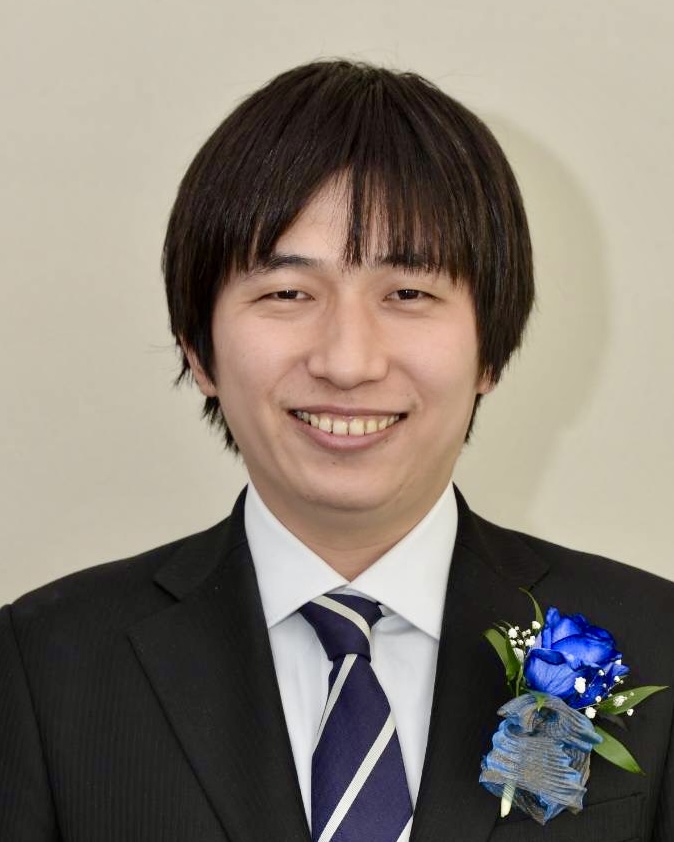}}]{Ichiro Maruta} (Member, IEEE) received the Bachelor of Engineering, Master of Informatics, and Doctor of Informatics degrees from Kyoto University, Kyoto, Japan, in 2006, 2008, and 2011, respectively. He was a Research Fellow with the Japan Society for the Promotion of Science, from 2008 to 2011. From 2012 to 2017, he was an Assistant Professor with the Graduate School of Informatics, Kyoto University. In 2017, he joined the Department of Aeronautics and Astronautics, Graduate School of Engineering, Kyoto University, as a Lecturer, where he has been an Associate Professor, since 2019.
\end{IEEEbiography}

\begin{IEEEbiography}[{\includegraphics[width=1in,height=1.25in,clip,keepaspectratio]{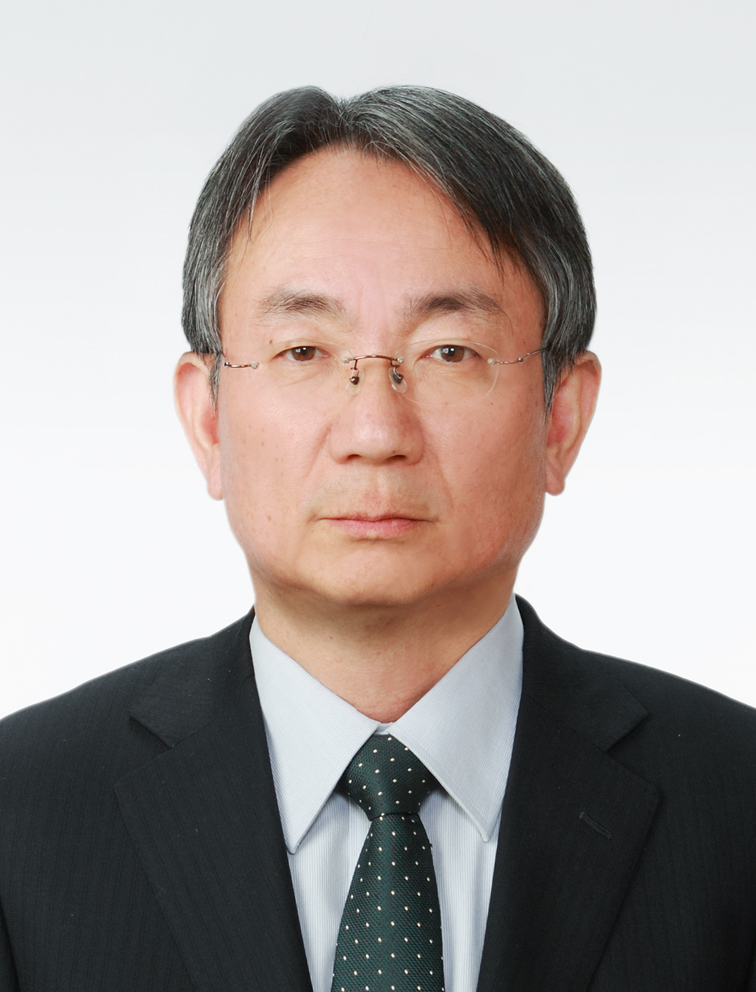}}]{Toshiharu Sugie} (Life Fellow, IEEE) received the B.E., M.E., and Ph.D. degrees in engineering from Kyoto University, Japan, in 1976, 1978, and 1985, respectively. From 1978 to 1980, he was a Research Member of the Musashino Electrical Communication Laboratory, NTT, Musashino, Japan. From 1984 to 1988, he was a Research Associate with the Department of Mechanical Engineering, University of Osaka Prefecture, Osaka. From 1988 to 2019, he worked at Kyoto University, where he has been a Professor with the Department of Systems Science, since 1997. In 2019, he joined Osaka University as a Research Member of the Komatsu MIRAI Construction Equipment Cooperative Research Center, Graduate School of Engineering. His research interests include robust control, identification for control, and control application to mechanical systems. He served as an Editor for Automatica. He was also an Associate Editor of Asian Journal of Control and the International Journal of Systems Science.
\end{IEEEbiography}

\end{document}